\pdfoutput=1

\documentclass[a4paper]{article}

\usepackage[utf8]{inputenc}
\usepackage{capt-of}
\usepackage{quiver}
\usepackage{xcolor}
\usepackage{pdfpages}
\usepackage{comment}
\usepackage{hyperref}
\usepackage{amsmath}
\usepackage{amsthm}
\usepackage{amssymb}

\usepackage{enumerate}
\usepackage{multicol}
\usepackage{chngcntr}
\usepackage{array}
\usepackage{xspace}
\usepackage{wrapfig}

\newcommand{\fullv}[1]{#1} 
\newcommand{\confv}[1]{} 
\newcommand{\condInc}[1]{\fullv{#1}} 

\newtheorem{theorem}{Theorem}[section]

\newtheorem{lemma}[theorem]{Lemma}
\newtheorem{definition}[theorem]{Definition}
\newtheorem{example}[theorem]{Example}

\newcommand{\Gr}{G}
\newcommand{\Gmax}{\Gr-\text{max}}

\newcommand{\rmap}{\vartheta_G}
\newcommand{\qmap}{(\Gr,\rmap)}

\newcommand{\br}[1]{\overline{#1}}

\newcommand{\bsim}{\ \sim\ }

\newcommand{\join}{\lor}
\newcommand{\meet}{\land}

\newcommand{\spec}{f}

\counterwithin{equation}{section}

\newcommand{\sbst}{\le} 

\newcommand{\M}{\mathcal{M}}  
\newcommand{\Mp}{\mathcal{M'}}
\newcommand{\Mpp}{\mathcal{M''}}
\newcommand{\N}{\mathcal{N}}
\newcommand{\Np}{\mathcal{N'}}

\newcommand{\St}{\ensuremath{S}}  
\newcommand{\Stp}{\ensuremath{S'}}

\newcommand{\MSt}{\ensuremath{\M_{\St}}}  
\newcommand{\MpSt}{\ensuremath{\M'_{\St}}}

\newcommand{\iSt}{\ensuremath{S_0}}

\newcommand{\MiSt}{\ensuremath{\M_{\iSt}}}
\newcommand{\MpiSt}{\ensuremath{\M'_{\iSt}}}

\newcommand{\Tr}{\ensuremath{T}}  

\newcommand{\MTr}{\ensuremath{\M_\Tr}}
\newcommand{\MpTr}{\ensuremath{\M'_\Tr}}
\newcommand{\NTr}{\ensuremath{\N_\Tr}}

\newcommand{\Lb}{\ensuremath{L}}  

\newcommand{\MLb}{\ensuremath{\M_{\Lb}}} 
\newcommand{\MpLb}{\ensuremath{\M'_{\Lb}}} 
\newcommand{\NLb}{\ensuremath{\N_{\Lb}}} 

\newcommand{\AP}{\ensuremath{\ensuremath{AP}}}

\newcommand{\MAP}{\ensuremath{\ensuremath{\M_{\AP}}}}
\newcommand{\MpAP}{\ensuremath{\ensuremath{\M'_{\AP}}}}

\newcommand{\Mq}{\br{\M}}
\newcommand{\MqSt}{\Mq_S}
\newcommand{\MqTr}{\Mq_T}

\newcommand{\Nq}{\br{\N}}
\newcommand{\sq}{\br{s}}
\newcommand{\tq}{\br{t}}

\newcommand{\kdef}{\ensuremath{\left( \St, \iSt, \Tr, \Lb, \AP \right)}\xspace}

\newcommand{\Prg}{P}     
\newcommand{\Prc}{P}     
\newcommand{\PrgR}{P'}   
\newcommand{\PrcR}{P'}    
\newcommand{\Ac}{AC}     
\newcommand{\action}[1]{action{(#1)}}  
\newcommand{\NR}{\mathcal{N'}}
\newcommand{\NRSt}{\mathcal{N'}_{\St}}

\newcommand{\CTL}{{\rm CTL}}
\newcommand{\ctlmod}[1]{\textup{\textsf{#1}}}   

\newcommand{\true}{\ensuremath{\mathsf{true}}\xspace}
\newcommand{\false}{\ensuremath{\mathsf{false}}\xspace}

\newcommand{\A}{\ctlmod{A}}
\newcommand{\E}{\ctlmod{E}}

\newcommand{\U}{\ctlmod{\,U\,}}

\newcommand{\V}{\ctlmod{\,R\,}}
\newcommand{\R}{\ctlmod{\,R\,}}

\newcommand{\AG}{\ctlmod{AG}}
\newcommand{\EG}{\ctlmod{EG}}
\newcommand{\AF}{\ctlmod{AF}}
\newcommand{\EF}{\ctlmod{EF}}

\newcommand{\AV}{\ctlmod{AR}}
\newcommand{\EV}{\ctlmod{ER}}

\newcommand{\AX}{\ctlmod{AX}}
\newcommand{\EX}{\ctlmod{EX}}

\renewcommand{\tt}{\ensuremath{\mathit{tt}}}     
\newcommand{\ff}{\ensuremath{\mathit{ff}}}

\newcommand{\vphi}{\varphi}

\newcommand{\AND}{\bigwedge}

\newcommand{\PA}{/\!/}		

\newcommand{\ar}{\rightarrow}	

\newcommand{\fa}{\forall}

\newcommand{\ind}{\hspace*{2em}}

\newcommand{\la}[1]{\mbox{$\, \stackrel{#1}{\rightarrow} \,$}}

\newcommand{\pl}{\!\parallel\!}
\newcommand{\pj}{\ensuremath{\!\upharpoonright\!}}

\newcommand{\sat}{\models}
\newcommand{\set}[1]{\ensuremath{\{#1\}}}
\newcommand{\stt}{~|~}

\newcommand{\sub}{\subseteq}
\newcommand{\un}{\cup}

\newcommand{\ie}{i.e.,\xspace}
\newcommand{\eg}{e.g.,\xspace}
\newcommand{\wrt}{w.r.t.\xspace}
\newcommand{\etal}{et.~al.\xspace}
\newcommand{\remove}[1]{}

\title{Model Repair via Symmetry}

\author{Paul Attie\footnote{School of Computer and Cyber Sciences, Augusta University, Augusta GA}\,\,\, and William Cocke\footnote{Correspoding author, School of Computer and Cyber Sciences, Augusta University, Augusta GA, wcocke@augusta.edu}}

\begin{document}
\maketitle

\begin{abstract}
The symmetry of a Kripke structure $\M$ has been exploited to replace a
model check of $\M$ by a model check of the potentially smaller structure
$\N$ obtained as the quotient of $\M$ by its symmetry group $\Gr$.
We extend previous work to model repair: identify a substructure that
satisfies a given temporal logic formula. We show that the
substructures of $\M$ that are preserved by $\Gr$ form a lattice that maps to the substructure lattice of $\N$.
We also show the existence of a monotone Galois
connection between the lattice of substructures of $\N$ and the
lattice of substructures of $\M$ that are "maximal"
w.r.t. an appropriately defined group action of $\Gr$ on $\M$.
These results enable us to repair $\N$ and then to lift the repair to $\M$. We can thus repair symmetric finite-state concurrent
programs by repairing the corresponding $\N$, thereby effecting program
repair while avoiding state-explosion.
\end{abstract}
\section{Introduction}

The model checking \cite{CES86} problem for (finite-state) concurrent programs 
$\Prg = \Prg_1 \cdots \pl \Prg_n$ \wrt a temporal logic formula $\spec$ is to verify that the Kripke structure $\M$ that is generated by the execution of $\Prg$ satisfies $\spec$. 
A major obstacle has been \emph{state explosion}: in general, the size of $\M$ is exponential in the number of processes $n$.
When $\M$ and $\spec$ both have a high degree of symmetry in the process index set $\set{1,\ldots,n}$, the use of \emph{symmetry reduction} to ameliorate state-explosion  can yield significant reduction in the complexity of model checking $\M \sat \spec$.

An approach that has been thoroughly investigated is the computation of the quotient structure $\Mq = \M/\Gr$ of $\M$ by the symmetry group $\Gr$ of symmetries of both $\M$ and $\spec$. $\Mq$ can be computed directly from $\Prg$, thereby avoiding the expensive computation of the large structure $\M$. Since model checking 
$\Mq \sat \spec$ is linear in the size of $\Mq$
\cite{CES86}, this provides significant savings if 
$\Mq$ is small. A related approach is the use of structural methods to express symmetric designs, \eg parameterized systems, where processes are all instances of a common template (possibly with a distinguished controller process), and rings of processes, where all communication is between a process and its neighbors in the ring.

\subsection{Our contributions} 
We present a theory of substructures of Kripke structures which provides a correspondence theorem between certain substructures of the Kripke structure $\M$ and the substructures of the quotient structure $\Mq$. Substructures of $\M$ correspond to possible repairs of $\M$, while substructures of $\Mq$ correspond to repairs of $\Mq$. Because the correspondence between the two lattices of substructures preserves semantic valuation of certain formulae, we can associate certain repairs of $\Mq$ with repairs of $\M$. Syntactically the correspondence of Kripke substructures lattices is of independent mathematical interest as an example of a monotone Galois connection.

We then build on this theory to extend symmetry-based model checking to 
\emph{concurrent program repair}: given a concurrent program $\Prg$
that may not satisfy $\spec$, modify $\Prg$ to produce a program that does satisfy $\spec$.
Given $\Prg$, $\spec$, and a group $\Gr$ of symmetries of both $\Prg$ and $\spec$, our method directly computes the quotient $\M/\Gr$ (following \cite{EmSi96}), then repairs $\M/\Gr$, using the algorithm of \cite{ABS18}, and finally, extracts a correct program from the repaired structure. 
\section{Related Work}

As discussed above, group theoretic approaches to symmetry-reduction  \cite{EN95,EmSi96,ES97,CEJS98,ET98,ET99,EHT00,EW05,CJEF96}
compute the quotient $\M/\Gr$ and model check that, instead of the original (much larger) structure $\M$. A key consideration is to devise a symmetry group $\Gr$ that is as large as possible, so as to achieve the greatest possible reductions. 
Parametrized systems \cite{EN96,EN98,EK00,EK02,EN03,EK03,EK03c,CGJ97,JJW0S20,CTV06} are those which consist of an arbitrary number of isomorphic instances of a single process template. The usual approach is to establish a \emph{cutoff}: a small value $n$ such that correctness of a system with $n$ instances implies correctness of arbitrarily large systems.
Special subclasses of parametrized systems include 
token-passing systems \cite{AJKR14,CTTV04}, which  circulate one or more tokens amongst isomorphic processes, and systems based on rendezvous/broadcast/multicast \cite{GS92,EFM99,EK03b}. 
Bloem \etal \cite{ParametrizedSurvey15} surveys the decidability of verification problems for parametrized systems.

Buccafurri \etal \cite{BEGL99} posed the repair problem for CTL and solved it
using abductive reasoning to generate repair suggestions that are verified by
model checking. 
Jobstmann \etal \cite{JGB05} and Staber \etal \cite{SJB05}
present game-based repair methods for programs and circuits.
Their method is complete for invariants only.
%
Carrillo and Rosenblueth \cite{CR09} presents a syntax-directed repair method that uses
``protections'' to deal with conjunction, and a representation of
Kripke structures as ``tables''. 
%
Chatzieleftheriou \etal \cite{CBSK12} repair abstract structures, using Kripke modal
transition systems and 3-valued CTL semantics.
von Essen and Jobstmann \cite{EJ15} present a game-based repair method which attempts to keep the repaired program close to the original faulty program, by 
also specifying a set of traces that the repair must 
leave intact.

\section{Preliminary Definitions}
In this section we include the formal definitions necessary for structure repair via symmetry. The section is broken into two parts: the first part covers the basic definitions of Kripke structures and CTL as can be found in \cite{EC82}; the second part covers the basic connections with group theory.

\subsection{Kripke Structures and CTL}
\label{sec:ctl}
Given a set $\AP$ of atomic propositions we define the propositional
branching-time temporal logic \CTL \cite{Eme90,EC82} by the
following grammar:
 \[
 \vphi ::= \true \mid \false \mid p \mid \neg \vphi \mid \vphi \land \vphi \mid \vphi \lor \vphi \mid 
           \AX \vphi \mid \EX \vphi \mid \A[\vphi \V \vphi] \mid \E[\vphi \V \vphi]
 \]
where $p \in \AP$, and $\true, \false$ are constant propositions whose
interpretation is the semantic truth values $\tt, \ff$ respectively.
A propositional formula is a \CTL formula not containing the $\AX$, $\EX$, $\AV$, $\EV$ modalities.
%
The semantics of $\CTL$ are defined with respect to a Kripke structure:

\begin{definition}[Kripke structure]
\label{defn:structure}
\sloppy A \textbf{Kripke structure} $\M$ is a tuple $\kdef$ where $\St$ is a finite nonempty set of states, $\iSt \sub \St$ is a nonempty set of initial states, 
$\Tr \subseteq (\St \times \St)$ is a 
transition relation, $\AP$ is a set of atomic propositions, and 
$\Lb: \St \ar 2^{\AP}$ is a labeling function that associates each state $s \in \St$ with a subset of atomic
propositions, namely those that hold in state $s$.  
We require that $\M$ be total, i.e., 
$\forall s \in \St, \exists t \in \St: (s,t) \in \Tr$. 
\end{definition}

When referring to the constituents of $\M = \kdef$, we write 
$\MSt, \MiSt, 
\MTr, \MLb,$ and $\MAP$ respectively. 
State $t$ is a \emph{successor} of state $s$ in $M$ iff $(s,t) \in \Tr$.
 A path $\pi$ in $\M$ is a (finite or infinite) sequence
of states, $\pi =s_{0}, s_{1}, \ldots$, such that $\forall i \geq 0: (s_i,s_{i+1}) \in \Tr$.  A
fullpath is an infinite path. 
State $t$ is reachable from state $s$ iff there is a path from $s$ to $t$.
State $t$ is reachable iff there is a path from an initial state to $t$.
Without loss of generality, we assume that a Kripke structure does not
contain unreachable states, \ie every $s \in \St$ is reachable.

Throughout the rest of the paper we will refer to two running examples of Kripke structures. We will refer to these as Example 2-mutex and Example Box. 

\begin{example}[$2$-mutex]
\label{ex2mutex}
This is the classic $2$-mutex example. Our Kripke structure corresponds to two processes acting concurrently. Each process can be in one of 3 states, either Neutral, Trying, or Critical and can transiting from Neutral to Trying, from Trying to Critical, and from Critical to Neutral. 

We represent this Kripke structure diagrammatically in the left hand side of Figure \ref{fig:mutex2}. We see all of the transitions represented by arrows, the initial states are shaded green, and the labels can be seen in each state. 
\end{example}

\begin{example}[Box]\label{exbox}
The so-called Box Kripke structure in Figure~\ref{fig:box} is an example 
of a highly symmetric structure, as we will discuss in section \ref{sec:sym}. This structure has 4 states and is represented in Figure \ref{fig:box}. Of note, we assign all states with the same label, making Box a boring structure with respect to semantic evaluation. However, the simplicity of Box will be useful when considering the syntactic concept of a group of state mappings and the various substructures of Box. 
\end{example}

\begin{wrapfigure}{r}{2in}
    \centering
    \includegraphics[scale = .2]{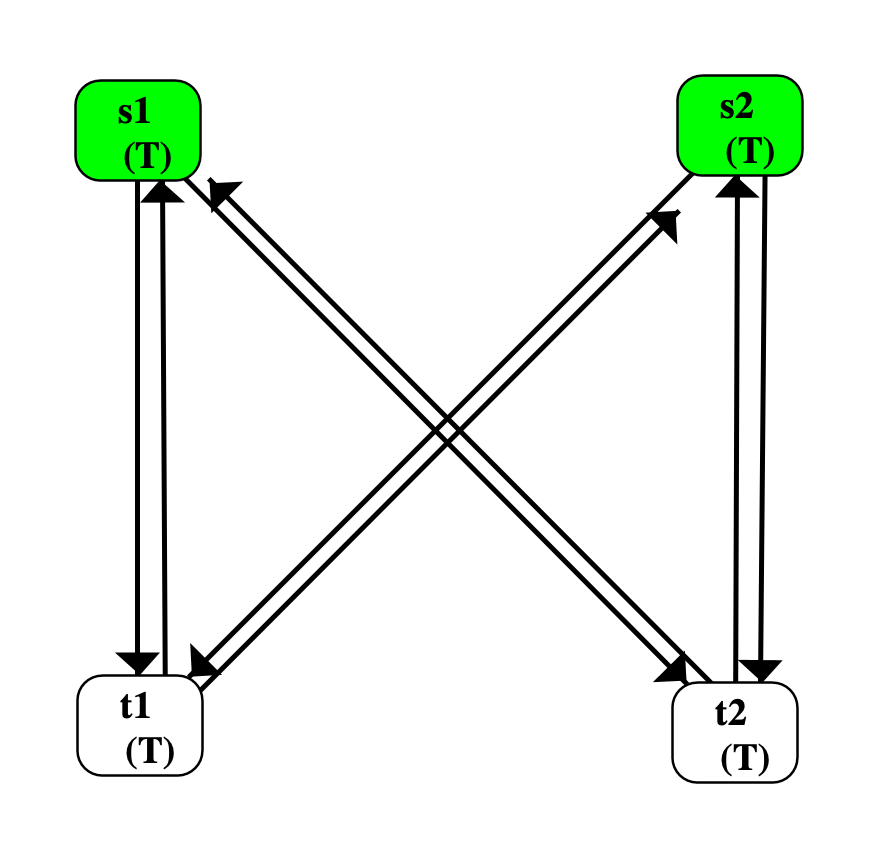}
    \caption{The Box Kripke structure. Useful for visualizing concepts.}
    \label{fig:box}
\end{wrapfigure}


A CTL formula $\vphi$ is evaluated (\ie is true or false) in a state $s$ of a Kripke structure $\M$ \cite{EC82}. Its value in $s$ depends on the value of its subformulae in states that are reachable from $s$. 
Write $\M, s \sat \vphi$ when $s$ is true in state $s$ of structure $\M$, and write $\M \sat \vphi$ to abbreviate $\fa s_0 \in S_0: \M, s_0 \sat \vphi$,
\ie $\vphi$ holds in all initial states of $\M$.
Appendix~\ref{app:ctl} gives the formal definition of $\sat$, which proceeds by induction on the structure of CTL formulae.

$[\vphi \V \psi]$ ($\vphi$ ``releases'' $\psi$) means that $\psi$ must hold up to and
including the first state where $\vphi$ holds, and forever if $\vphi$ never holds.
We introduce the abbreviations 
$\A[\phi \U \psi]$ for $\neg \E[\neg \vphi \V \neg \psi]$,
$\E[\phi \U \psi]$ for $\neg \A[\neg \vphi \V \neg \psi]$,
$\AF \vphi$ for $\A[\true \U \vphi]$,
$\EF \vphi$ for $\E[\true \U \vphi]$,
$\AG \vphi$ for $\A[\false \V \vphi]$,
$\EG \vphi$ for $\E[\false \V \vphi]$.
We use $\tt$, $\ff$ for the (semantic) truth values of $\true$, $\false$,
respectively. Whereas $\true$, $\false$ are atomic propositions whose
interpretation is always $\tt$, $\ff$, respectively.


\section{Substructures of Kripke structures}

Symmetry-based repairs are motivated by the correspondence between certain substructures of two related Kripke structures which we introduce as Theorem \ref{thm:correspondence}. A substructure of a Kripke structure is defined as follows:

\begin{definition}[Substructure]
\label{defn:substructure}
Given Kripke structures $\M$ and $\Mp$, we say that 
$\Mp$ is a \textbf{substructure} of $\M$, denoted 
$\Mp \sbst \M$, iff the following all hold:
\begin{enumerate}
    \item $\MpSt \sub \MSt$.
    \item $\MpiSt = \MiSt \cap \MpSt$. 
    \item $\MpTr \sub \MTr$. 
    \item $\MpAP = \MAP$. 
    \item $\MpLb = \MLb \pj \Stp$ (where
$\pj$ denotes domain restriction).
\end{enumerate}
\end{definition}


We require substructures to be total. To reiterate, a substructure $\N$ of a Kripke structure $\M$ is \textbf{total}. For mathematical necessity in what follows, we allow for the `empty' substructure. Since substructures must be total, some Kripke structures contain only two substructures, i.e., the empty substructure and the whole structure. 

We return briefly to our examples.

\begin{example}[\textbf{Example 2-mutex}]
Let $\M$ be the Kripke structure correspond to our running example 2-mutex (see Figure \ref{fig:mutex2}).

We can form a number of substructures of $\M$. For example, there is a substructure of $\M$ consisting only of states $s_0, s_1, s_2, s_3$ and $s_5$ along with all edges from $\M$. This substructure corresponds to computations wherein each process must return to neutral before the other process can enter the trying state. The states $s_0, s_1, s_3$ also form a substructure. 
\end{example}

\begin{example}[\textbf{Example Box}]
Like our 2-mutex example, Example Box has many substructures. We present some of these substructures in Figure \ref{fig:box_sub}.  

\begin{figure}
    \centering
    \includegraphics[scale = .1]{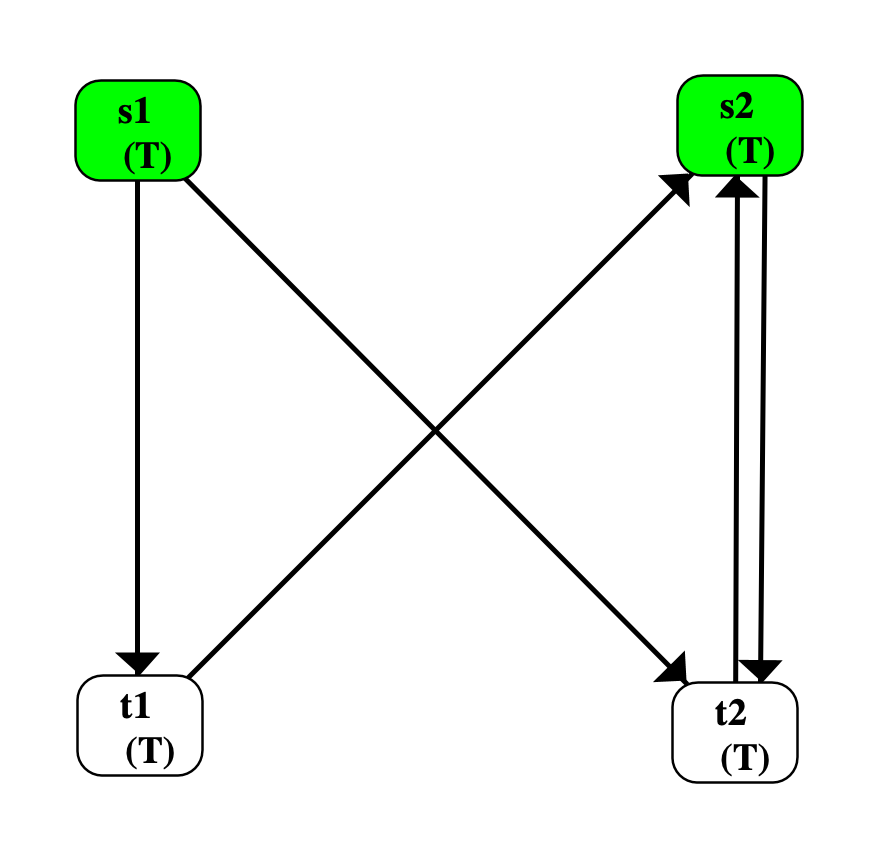}
    \includegraphics[scale = .1]{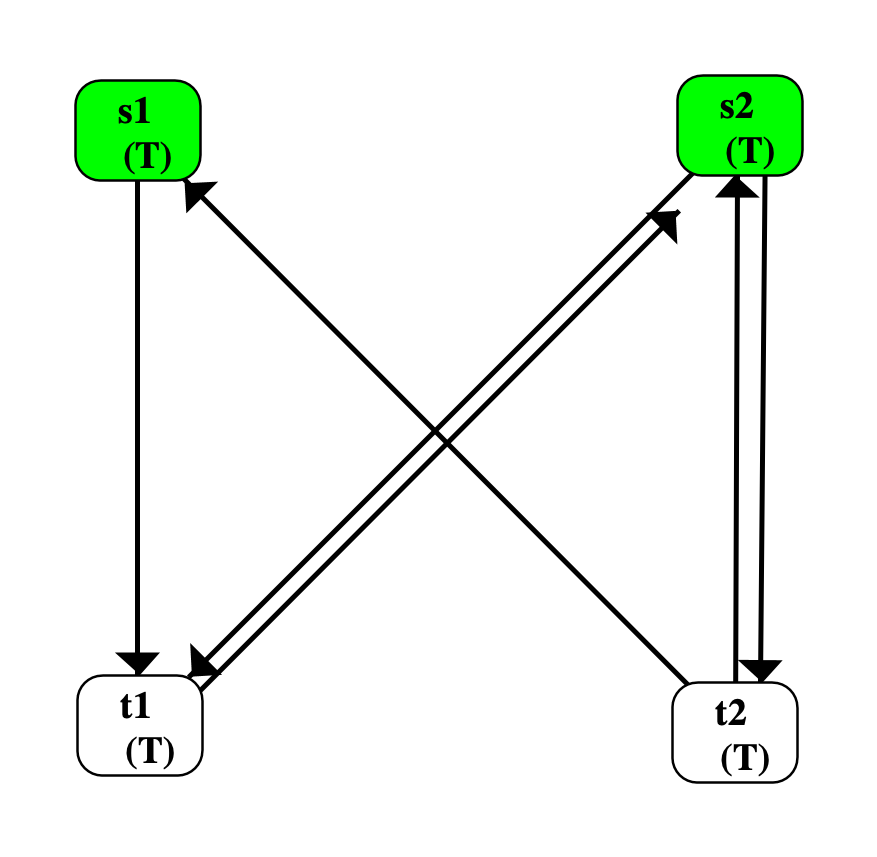}
    \includegraphics[scale = .1]{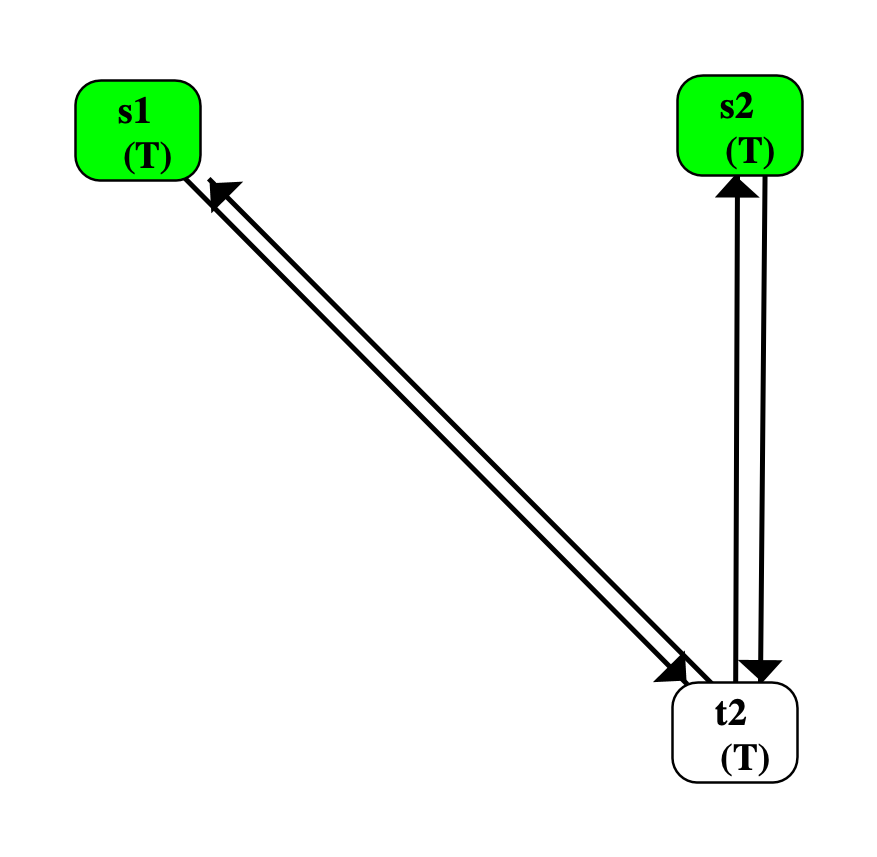}
    \caption{Some substructures of Example Box. There are many substructures not shown.}
    \label{fig:box_sub}
\end{figure}
\end{example}

As seen in the above examples, even a simple Kripke structure can have many different substructures, even on the same set of states. 

We can define a partially ordered set of substructures of a Kripke structure. This partially ordered set is a lattice i.e., it has a join and a meet operation. Before we define this lattice, we  investigate a few more basic properties regarding substructures of Kripke structures. 

Given a Kripke structure $\M$, a substructure $\Mp$ of $\M$, and a substructure $\Mpp$ of $\Mp$ it is immediate that $\Mpp$ is a substructure of $\M$. 

\begin{lemma}\label{lem:str_union}
Let $\mathcal{M}$ be a Kripke structure and suppose that $\mathcal{N}$ and $\mathcal{N}'$ are substructures of $\mathcal{M}$. Then 
\[\mathcal{N} \vee \mathcal{N}' = (\mathcal{N}_S \cup \mathcal{N}'_S,\, \mathcal{N}_{S_0} \cup \mathcal{N}'_{S_0},\, \mathcal{N}'_R \cup \mathcal{N}'_R,\, \mathcal{M}_L \pj (\mathcal{N}_S \cup \mathcal{N}'_S),\, \mathcal{M}_X)\] is the smallest substructure of $\mathcal{M}$ containing both $\mathcal{N}$ and $\mathcal{N}'$. 
\end{lemma}

\begin{proof}
We will write $\mathcal{N}_\vee $ for $\mathcal{N} \vee \mathcal{N}'$. By construction $\mathcal{N}_\vee$ contains both $\mathcal{N}$ and $\mathcal{N}'$. Let $\mathcal{T}$ be any substructure of $\mathcal{M}$ containing both $\mathcal{N}$ and $\mathcal{N}'$. Clearly $\left(\mathcal{N}_\vee
\right)_S \subseteq \mathcal{T}_S$, $\left(\mathcal{N}_\vee\right)_{S_0} \subseteq \mathcal{T}_{S_0}$, $\mathcal{T}_R \subseteq \mathcal{T}_R$, and $\mathcal{T}_L$ extends $\left(\mathcal{N}_\vee\right)_L$. 

It remains to show that $\mathcal{N}_\vee$ is a substructure, i.e., that $\mathcal{N}_\vee$ is total. Given a state $s_1$ in $\left(\mathcal{N}_vee\right)_S$, we know that either $s\in \mathcal{N}_S$ or $s\in \mathcal{N}_S'$. If $s\in \mathcal{N}_S$, then there is an $t\in \mathcal{N}_S$ such that $(s,t)\in \mathcal{N}_\Tr \subseteq \left(\mathcal{N}_\vee\right)_\Tr$. By symmetry if $s\in \mathcal{N}_S'$ then there is an $t\in \mathcal{N}_S'$ with $(s,t)\in \left(\mathcal{N}_\vee\right)_\Tr$. 
\end{proof}

Given a finite set $X=\{X_0,X_1,\dots,X_n\}$ of substructures of $\mathcal{M}$, we can define the structure $\bigvee X = X_0 \vee X_1 \vee \dots \vee X_n$.

\begin{lemma}
Let $\M$ be a Kripke structure and suppose that $\mathcal{N}$ and $\mathcal{N}'$ are substructures of $\mathcal{M}$. Then there is a largest substructure of $\mathcal{M}$ contained in both $\mathcal{N}$ and $\mathcal{N}'$.
\end{lemma}

\begin{proof}
The empty structure is contained in both $\mathcal{N}$ and $\mathcal{N}'$. Let $X$ be the collection of structures contained in both $\mathcal{N}$ and $\mathcal{N}'$. Then $\bigvee X$ is contained in both $\mathcal{N}$ and $\mathcal{N}'$ and is the largest substructure of $\mathcal{M}$ contained in both $\mathcal{N}$ and $\mathcal{N}'$. 
\end{proof}

\begin{definition}[Join, Meet of substructures]
Let $\mathcal{N}'$ and $\mathcal{N}'$ be two substructures of $\mathcal{M}$. The \textbf{join} of $\mathcal{N}$ and $\mathcal{N}'$, written $\mathcal{N} \vee \mathcal{N}'$, is the smallest substructure of $\mathcal{M}$ containing both $\mathcal{N}$ and $\mathcal{N}'$. The \textbf{meet} of $\mathcal{N}$ and $\mathcal{N}'$, written $\mathcal{N} \wedge \mathcal{N}'$, is the largest substructure of $\mathcal{M}$ contained in both $\mathcal{N}$ and $\mathcal{N}'$. 
\end{definition}

We note that for two substructures $\N$ and $\Np$ of $\M$, the join $\N \join \Np$ has the simple description as given in Lemma \ref{lem:str_union}. However, the meet $\N \meet \Np$, while well-defined, does not have such a simple description. It is possible that for two substructure $\N$ and $\Np$ of a Kripke structure $\M$, there are no substructures contain in both $\N$ and $\Np$. Hence the largest substructure contained in both $\N$ and $\Np$ could be empty.

We can now define a lattice of substructures $\Lambda_\mathcal{M}$ for a given structure $\mathcal{M}$. 

\begin{definition}[Lattice of Substructures]
\label{defn:lattice_substructures}
Given a Kripke structure $\mathcal{M}$ the \textbf{lattice of substructures of $\mathcal{M}$} is  
\[ \Lambda_\mathcal{M} = \left(\left\{\mathcal{N}: \mathcal{N} \text{ is a substructure of }  \mathcal{M} \right\}, \leq  \right),\] 
\noindent where the meet and join in $\Lambda_\mathcal{M}$ are the meet and join defined for substructures earlier. 
\end{definition}
There is another lattice of substructures related to $\M$. Given a subset $\Stp \subseteq \M_\St$ we can define the substructure $\N = \langle \Stp \rangle$ of $\M$ generated by $\Stp$ as the largest substructure $\Np$ of $\M$ such that $\Np_\St \subseteq \Stp$, equivalently  \[\N = \bigvee_{\substack{\Np \leq \M\\ \Np_\St \subseteq \Stp}} \Np.\] 

\section{Quotient structures.} \label{sec:sym}
Broadly speaking, a symmetry of an object captures the concept of what we mean for two objects to be equal. For Kripke structures, we are in general not interested in full symmetry. Instead, we are interested in a concept of partial symmetry that allows for certain computations to take place. The reason for this is because true symmetry of a Kripke structure would have to act on not only the underlying state-space, but the set of atomic propositions as well. We introduce a related concept, which we call a state-mapping. Some authors call a `state-mapping' a `symmetry' of $\mathcal{M}$. Due to the mathematical meaning of the word symmetry, we do not follow this tradition.

\subsection{Groups Acting on Kripke Structures}

\begin{definition}
A state-mapping of $\M$ is a graph isomorphism of the state-space of $\M$ such that its restriction to the initial states is also an isomorphism, e.g., takes initial states to initial states. Formally, for a Kripke structure $\M$, a \textbf{state-mapping} of $\M$ is a bijection $f:\MSt \rightarrow \MSt$ such that:
\begin{itemize}
    \item $f(\MiSt)= \MiSt;$
    \item For states $s,t \in \MSt$ we have that $(s, t) \in \MTr \iff (f(s),f(t)) \in \MTr$.
\end{itemize}
\end{definition}

The set of all state-mappings of $\M$ forms a group. 

\begin{definition}[$\Gr$-closed]
For a group $G$ of state-mappings of a Kripke structure $\M$, a substructure $\N$ of $\M$ is called \textbf{$G$-closed}, if $G$ is a group of state-mappings of $\N$, i.e., for every $g\in G$ and $s\in \N_{\St}$ we have $g(s) \in \N_{\St}$.
\end{definition}

\begin{lemma}\label{lem:g-closedlattice}
Let $\M$ be a Kripke structure and let $\Gr$ be a group of state mappings of $\M$. Let $\N, \Np$ be two $\Gr$-closed substructures of $\M$. Then 
$\N \join \Np$ and $\N \meet \Np$ are both $\Gr$-closed.
\end{lemma}

\condInc{
\begin{proof}
We first show that $\N\vee \N'$ is $G$-closed. Since $(\N\vee \Np)_\St = \N_\St \cup \Np_\St$, the group $G$ acts on the states of $\N\vee \Np.$ We will show that every element of $G$ is a state-mapping on $\N \vee \Np.$

Let $s,t \in (\N \vee \Np)_\St$. If $g(s) = g(t)$ then $s = t$ since $g$ is bijection on $\M_\St$. Similarly, for any $s \in (\N \vee \Np)_\St$, since $g^{-1}(s) \in (\N \vee \Np)_\St$. 

If $(s,t)\in (\N \vee \Np)_\Tr$, then $(s,t)$ is in one of $\N_\Tr$ or $\Np_\Tr$. In either case, $(g(s),g(t))$ is contained in the same set of transitions, and thus contained in $(\N \vee \Np)_\Tr$. 

We now consider $\N\land \Np$. Recall that $\N \wedge \Np$, is the largest substructure contained in both $\N$ and $\Np$. If $s \in (\N \wedge \Np)_\St$, and $g\in G$, then there is a substructure of $\N$ and $\Np$ containing $g(s)$. Moreover as the largest substructure of $\N$ and $\Np$, $\N \wedge \Np$ must contain this substructure, hence $g(s) \in \N\wedge \Np$. 

Consider $(s,t) \in (\N\land\Np)_\Tr$. Since $G$ is a group of state mappings of $\N \meet \Np$ there is a substructure of $\N$ and $\Np$ containing $(g(s),g(t))$. Since $\N\land \Np$ is the join of all substructures contained in $\N$ and $\Np$, we conclude that it must contain $(g(s),g(t))$ as well. 
\end{proof}
}

By Lemma \ref{lem:g-closedlattice}, we see that the $G$-closed substructures of $\M$ forms a sublattice of $\Lambda_\M$. 

\begin{definition}[Lattice of $G$-closed substructures] Given a Kripke structure $\M$ and a group $G$ of state mappings of $\M$, the poset of $G$-closed substructures of $\M$ forms a lattice. We call this lattice the \textbf{lattice of $G$-closed substructures of  $\M$} and write it as  $\Lambda_{\M,G}$. 
\end{definition}

\begin{example}[Example Box]\label{ex:box_gclosed}
Let $\M$ be Example Box, i.e., the Kripke structure described in Example \ref{exbox}. Let $g$ be the map that simultaneously switches $s_1$ and $s_2$, and switches $t_1$ and $t_2$, i.e. $g(s_1)=s_2, g(s_2) = s_1, g(t_1) = t_2, g(t_2) = t_1$. Let $G$ be the group consisting of $g$ and the identity map on $\MSt$. We note that $G$ is not the entire group of state-mappings of $\M$. The structure $\M$ has 10 $G$-closed substructures, including the empty structure. 
We present these structures in Figure \ref{fig:box_gclosed}.
\end{example}
 
\begin{figure}[th]
    \centering
    \includegraphics[scale = .12]{box.png}
    \includegraphics[scale = .12]{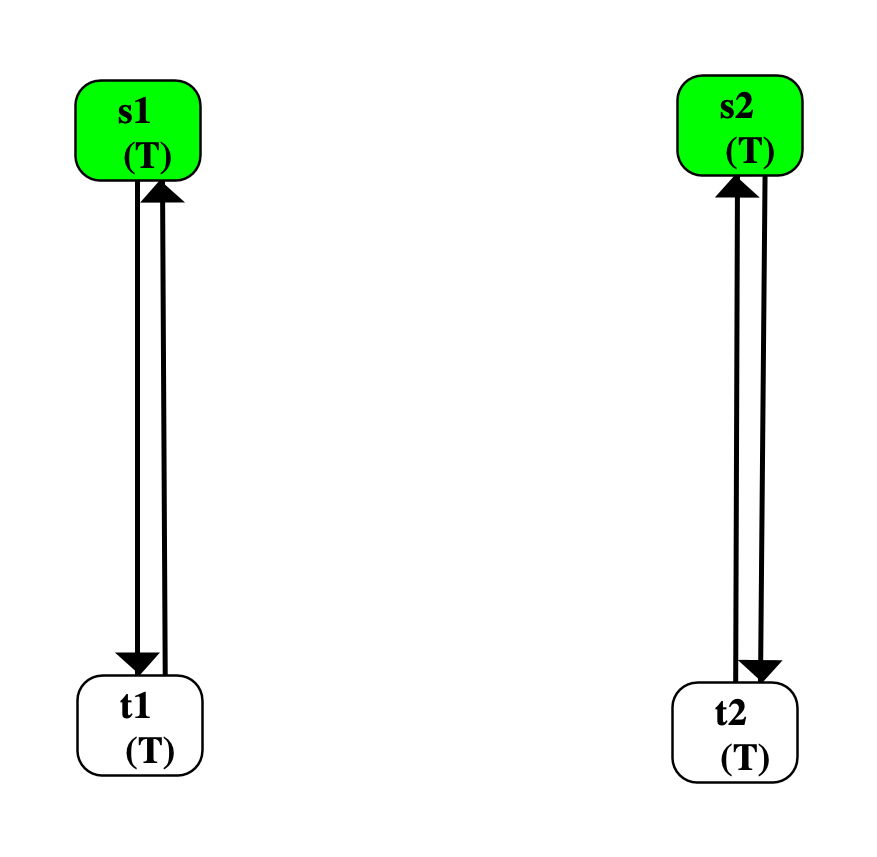}
    \includegraphics[scale = .12]{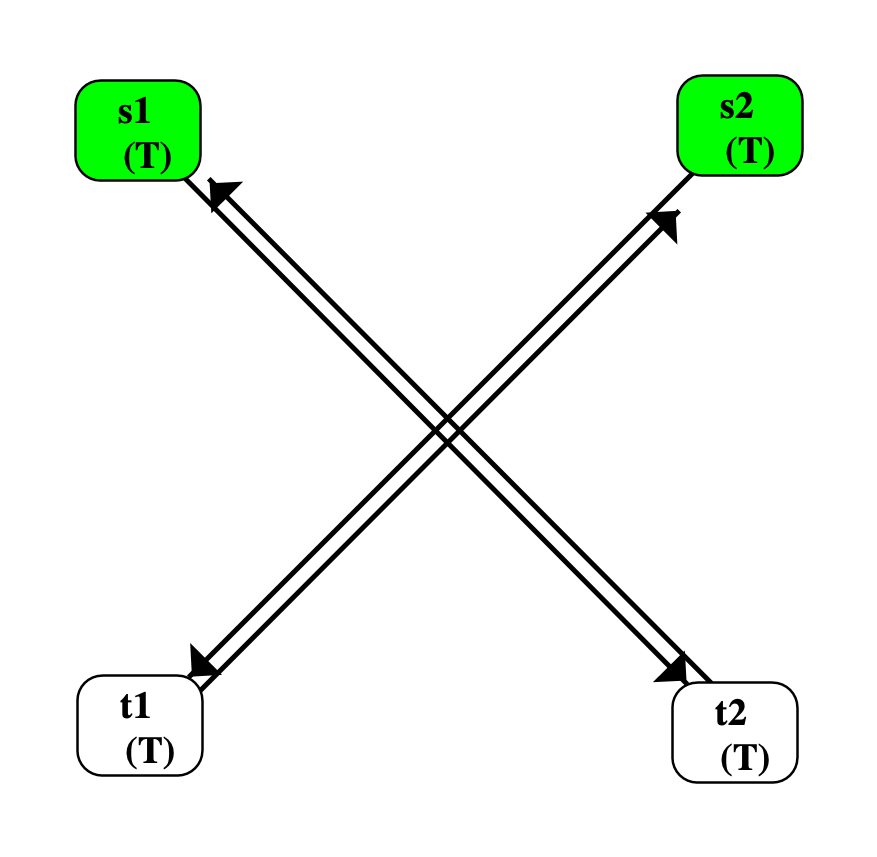}
    \includegraphics[scale = .12]{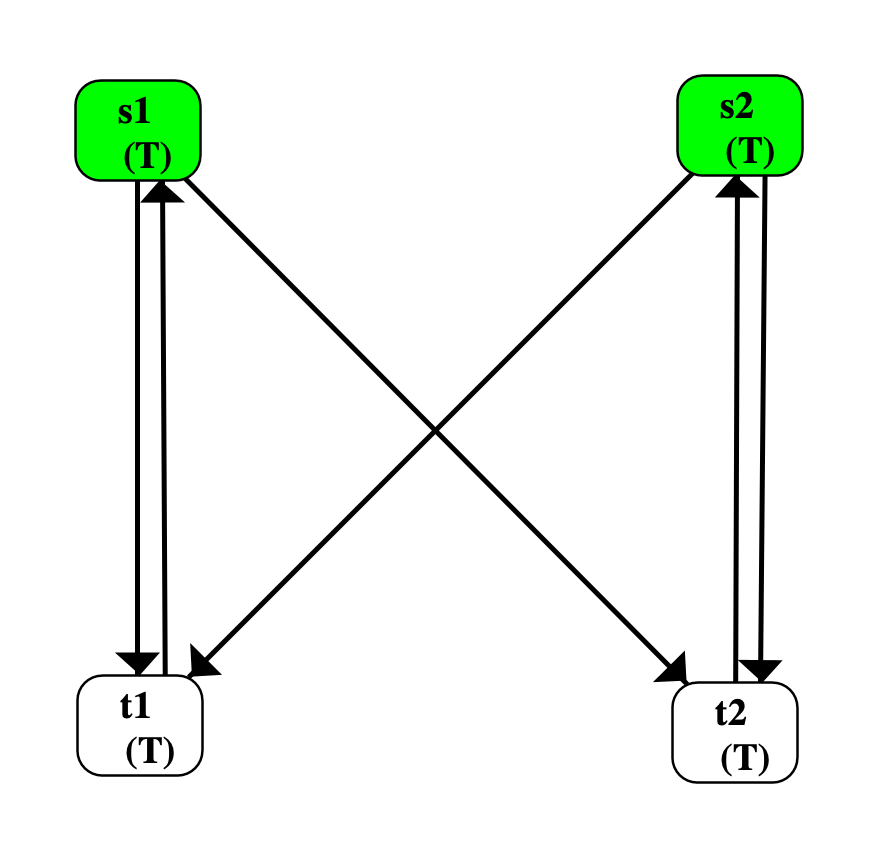}
    \includegraphics[scale = .12]{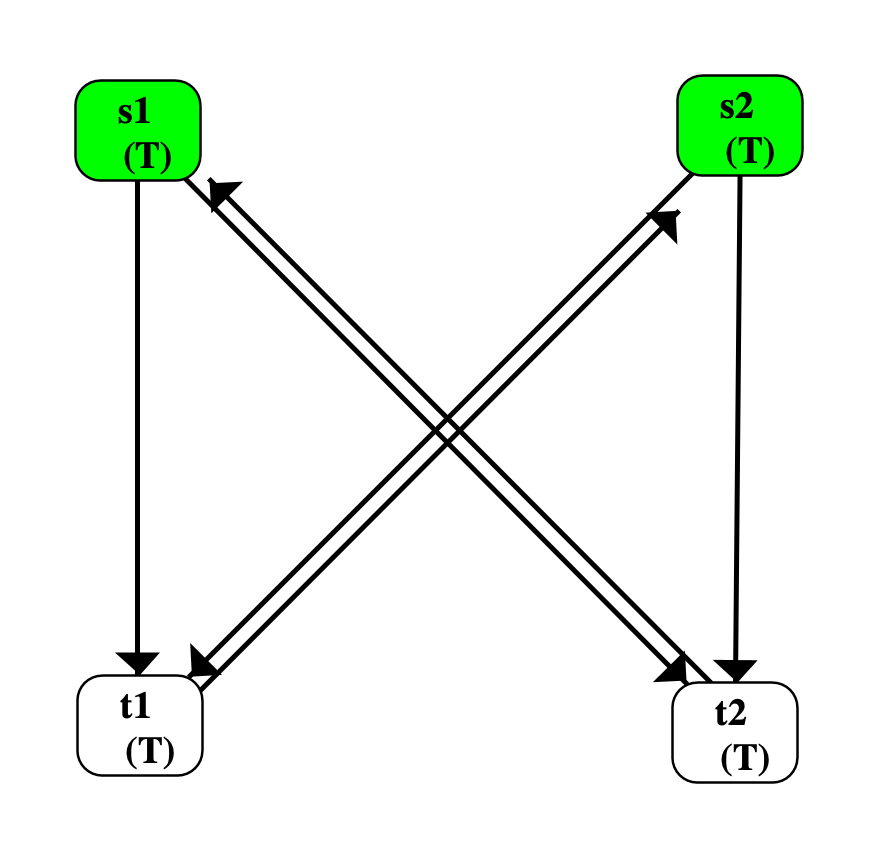}
    \includegraphics[scale = .12]{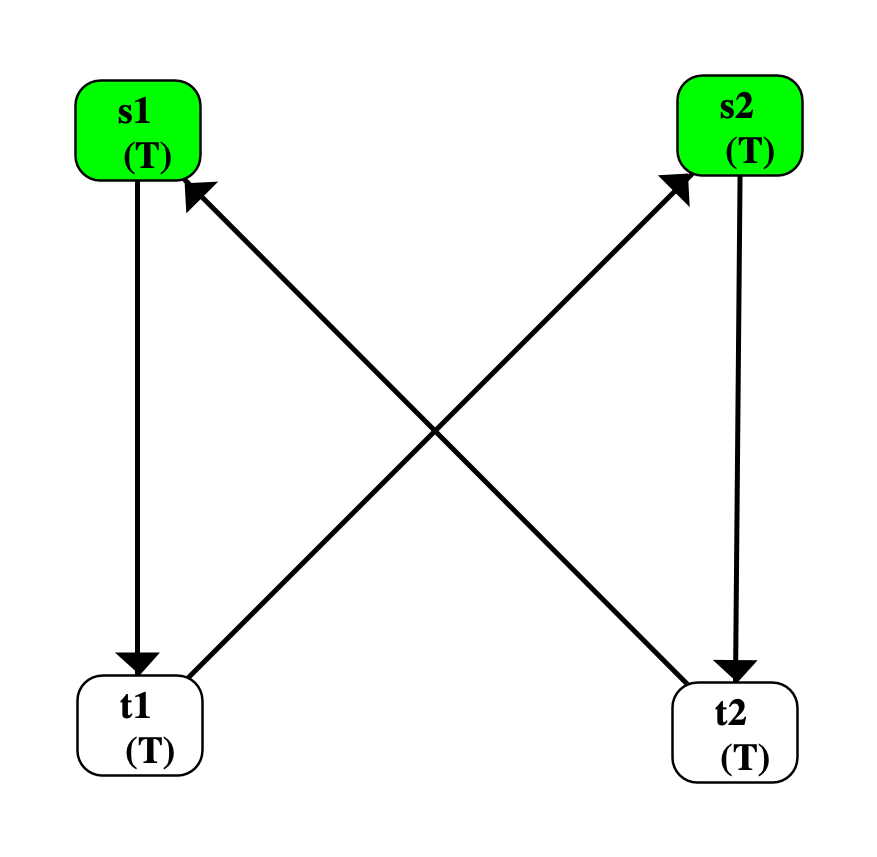}
    \includegraphics[scale = .12]{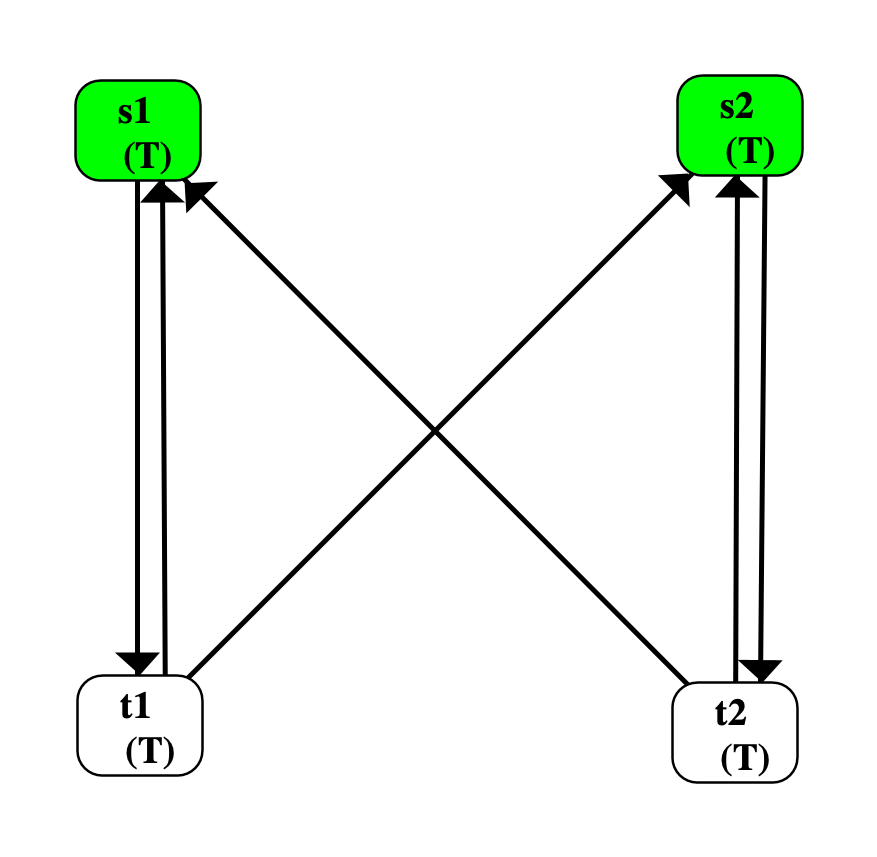}
    \includegraphics[scale = .12]{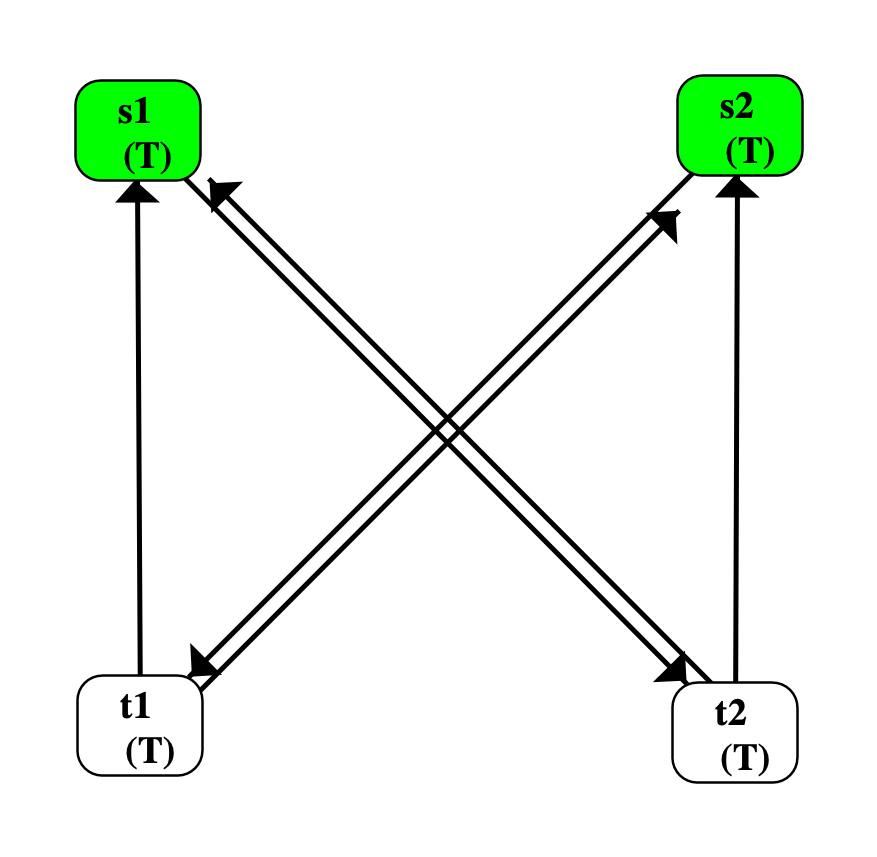}
    \includegraphics[scale = .12]{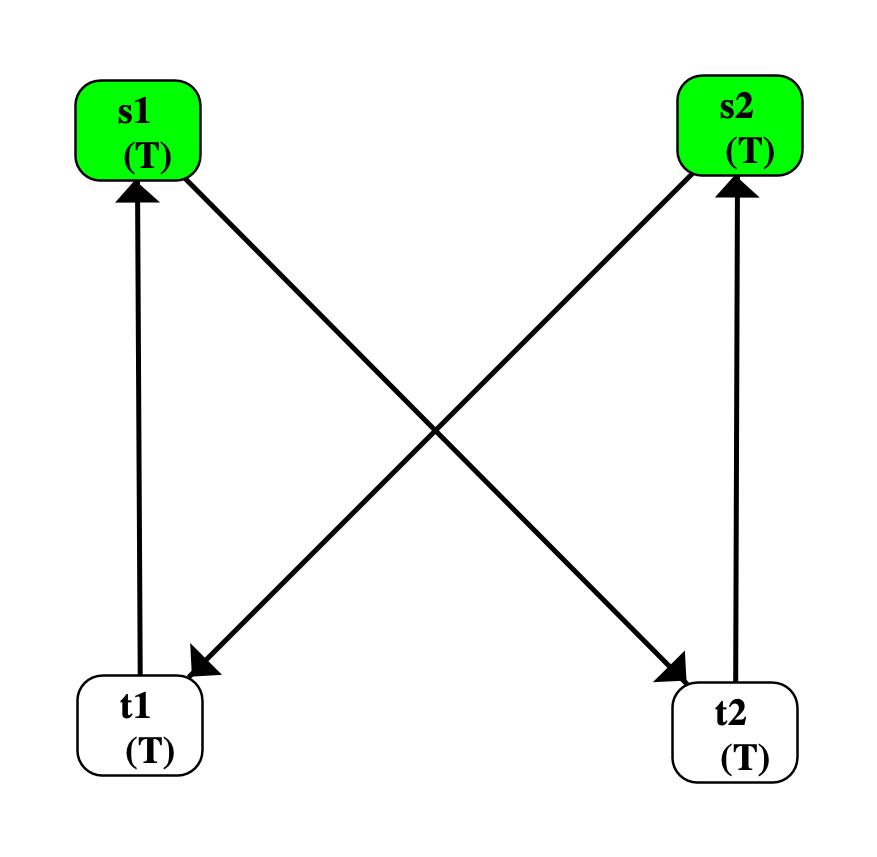}
    \caption{The nonempty $G$-closed substructure of Example Box. Where $G$ is the group generated by the simultaneous swapping of indexes of both the $s_i$ and the $t_i$.}
    \label{fig:box_gclosed}
\end{figure}

As evidenced by Figure \ref{fig:box_gclosed} in Example \ref{ex:box_gclosed}, a single Kripke structure can have many $G$-closed substructures. 

\subsection{Constructing the Quotient Model}

Give a group $G$ of state-mappings of a model $\M$, we want to construct a quotient model $\mathcal{M}_G$. However, as noted, state-mappings do not contain any information about $\mathcal{M}_L$. To remedy this situation, we need a function that assigns a representative to each orbit of $G$. 

\begin{definition}[Representative map]
Let $\M$ be a Kripke structure and suppose that $G$ is a group of state-mappings of $\mathcal{G}$. A \textbf{representative map} of $\M$ with respect to $G$ is a function $\rmap: \MSt \rightarrow \MSt$ satisfying the following:
\begin{itemize}
    \item For all $s, s' \in \MSt$, if there is a $g\in G$ such that $g(s) = s'$ then $\rmap(s) = \rmap(s')$. (respects orbits)
    \item For all $s, s' \in \MSt$, if there is no $g\in G$ such that $g(s) = s'$ then
    $\rmap(s) \neq \rmap(s')$. (separates orbits)
    \item For all $s\in \MSt$, we have that $\rmap(\rmap(s)) = \rmap(s)$, i.e., each orbit has a stable representative. (idempotent)
\end{itemize}
\end{definition}

We are now ready to define the quotient structure.

\begin{definition}[Quotient structure]
Given a Kripke structure $\M$, a group $\Gr$ of state-mappings of $\M$, and a representative map $\rmap$ of $\M$ with respect to $\Gr$, we define the \textbf{quotient structure} $\M/\rmap$ of $\M$ with respect to $\Gr$ and $\rmap$ as follows:

The set of states of $\M/\qmap$ consists of the image of $\rmap$. We borrow the quotient notation from group theory and write $\Mq$ for $\M/\qmap$ and $\sq$ for $\rmap(s)$. 

We have $(\sq,\tq)\in \MqTr$ if and only if there is some $s' \in \MSt$ with $\rmap(s') = \sq$ and some $t'\in \MSt$ with $\rmap(t') =  \tq$ such that $(s',t')\in \MTr$.

An initial state of $\Mq$ is the image under $\rmap$ of an initial state of $\M$.

We define $\overline{\mathcal{M}}_L:\overline{\mathcal{M}}_S\rightarrow 2^{AP}$ by $\overline{\mathcal{M}}_L(\sq) = \mathcal{M}_L(\rmap(s))$. 

$\overline{\M}_{\AP} = \MAP.$
\end{definition}

We note that the states of a quotient structure correspond exactly to the orbits of states of the original structure. For transitions, we have a slightly more subtle correspondence. 
Consider the following example:

\begin{example}[Example Box]\label{ex:box_quo}
\sloppy Let $\M$ and $G$ be as in Example \ref{ex:box_gclosed}. Let $\vartheta_G$ be defined by $\vartheta(x_i) = x_1$ for $x\in \{s,t\}$. Then the quotient structure $\M/(G,\vartheta_G)$ has exactly 2 states, $s_1 = \vartheta_G(s_1) = \vartheta_G(s_2)$ and $t_1 = \vartheta_G(t_1) = \vartheta_G(t_2)$ with transitions $(s_1,t_1),(t_1,s_1)$. Every $G$-closed substructure $\N$ of $\M$ in Figure \ref{fig:box_gclosed}, will map to this structure via \[\N \rightarrow \N/(G,\vartheta_G).\] Many of these substructures do no have $(s_1,t_1)$ as a transition at all. 
\end{example}

As shown in Example \ref{ex:box_quo}, the existence of a transition in the quotient structure, does not imply the existence of ``the same" transition in the original structure. However, it does imply the existence of a corresponding transition in the original structure. Hence it is true that for every path through a quotient structure there is a corresponding path through the original Kripke structure. This property is what allows for model checking of the original structure via model checking of the quotient. This path correspondence is formalized in the below theorem which has appeared in many texts. See for example the work by Emerson and Sistla \cite[\textbf{3.1}]{EmSi96}. 

\begin{theorem}[\textbf{Path Correspondence Theorem}]
\label{thm:path-correspondence}
Let $\M$ be a Kripke structure and let $G$ be a group of state-mappings of $\M$. Let $\rmap$ be any representative map of $\M$ with respect of $\Gr$. Write $\Mq$ for  $\M/\qmap$. There is a bidirectional correspondence between paths of  $\M$ and paths of  $\Mq$. Formally we have the following:
\begin{enumerate}
    \item If $x = s_0, s_1, s_2, \dots$ is a path in $\mathcal{M}$, then $\overline{x} = \overline{s_0}, \overline{s_1}, \overline{s_2}, \dots$ is a path in $\overline{\mathcal{M}}$ where $\overline{s_i} = \vartheta_G(s_i)$. 
    \item If $\overline{x} = \overline{s_0},\overline{s_1},\overline{s_2},\dots $ is a path in $\overline{\mathcal{M}}$, then for every state $s'_0 \in \mathcal{M}_S$ such that $\vartheta_G(s'_0) =  \overline{s_0}$ there is a path $s'_0, s'_1, s'_2,\dots$ such that $\vartheta_G(s'_i) = \overline{s_i}$. 
\end{enumerate}
\end{theorem}

\condInc{
\begin{proof}
1. Given $(s_i,s_{i+1})\in \mathcal{M}_R$, then $(\overline{s_i},\overline{s_{i+1}})\in \overline{\mathcal{M}}_R.$

\sloppy 2. Suppose that $(\overline{s_i},\overline{s_{i+1}})\in \overline{\mathcal{M}}_R$. Then there is some $(t,t')\in \mathcal{M}_R$ with $t\in \overline{s_i}$ and $t'\in \overline{s_{i+1}}.$ Suppose that we have a path $s_0',\dots,s_i'$ such that $s_0' \in \overline{s_0}$, $s_1'\in \overline{s_1},\dots, s_i'\in \overline{s_i}$ and $(s_0',s_1'),(s_1',s_2'),\dots,(s_{i-1}',s_i')\in \mathcal{M}_R.$  Since both $t,s_{i}' \in \overline{s_i}$, there is a group element $g\in G$ such that $g(t) = s_{i}'$. Then $g(s_0'),g(s_1'),\dots,g(s_i')=t,s_{i+1}$ is a partial path in $\mathcal{M}$ Continuing in this manner we see the existence of the desired path in $\mathcal{M}$.
\end{proof}
}

We would like to extend the path correspondence between $\mathcal{M}$ and $\mathcal{M}/\left(G, \vartheta_G\right)$, for a group $G$ of state-mappings of $\mathcal{M}$ and a representative map $\vartheta_G$ to one of substructures of Kripke structures. To do this, we need to identify a correspondence between $G$-closed substructures of $\M$ and substructures of $\M/(G,\vartheta_G)$. For $\N$ a $G$-closed substructure of $\M$, there is a natural map $\N\rightarrow \N/(G,\vartheta_G)$, which we call the quotient map. The quotient map establishes a lattice homomorphism between $\Lambda_{\M,\Gr}$ and $\Lambda_{\M/(G,\vartheta_G)}.$

We prove this property in a series of lemmas.

\begin{lemma}\label{lem:closed_are_surjective}
Let $\M$ be a Kripke structure and $G$ a group of state mappings of $\M$ with representative map $\vartheta_G$. Then for every substructure $\N$ of $\M/(G,\vartheta_G)$, there is a $\Gr$-closed substructure $\Np$ of $\M$ such that $\Np/(G,\vartheta_G) = \N$. 
\end{lemma}

\condInc{
\begin{proof}
We will construct a $G$-closed substructure $\Np$ of $\M$ such that $\Np/(G,\vartheta_G) = \N$. For the states of $\Np$, we select exactly the $G$-orbits of the representative states of $\N$, i.e., $\Np_\St = \{g(s): \forall g\in G, s\in \N_\St\}$. For transitions, we note that for every transition $(s,t)\in \N_\Tr$ there is at least one transition $(s',t')\in \M_\Tr$ such that $\vartheta(s')= s$ and $\vartheta(t') = t$. We add this transition and its orbit under $G$ to $\Np$. Hence for any state $s'\in\Np_\St$ there is a transition $(s',t')$ for some $t'\in \Np_\St$. We conclude that $\Np$ is total. Since $\Np$ is determined by its states and transitions as a substructure, the lemma is proved.
\end{proof}
}

Lemma \ref{lem:closed_are_surjective} establishes that the quotient map is surjective. The next lemma demonstrates that it is a homomorphism of the join-semilattices $\Lambda_{M,\Gr}$ and $\Lambda_{M/(G,\vartheta_G)}$. We note that it is not a homomorphism of the lattices themselves, because the meet of two $G$-closed structures mapping might be empty due to the requirement that Kripke structures are total. 

\begin{lemma}[Quotient map respects join]
Let $\M$ be a Kripke structure and $G$ a group of state mappings of $\M$ with representative map $\vartheta_G$. Let $\N,\Np \in \Lambda_{\M,\Gr}$. Then 
     \[(\N \join \Np)/(\Gr, \vartheta_\Gr) = \N/(\Gr,\vartheta_\Gr) \join \Np/(\Gr,\vartheta_\Gr).\]
\end{lemma}

\condInc{
\begin{proof}

Clearly, by definition a state $\vartheta_{\Gr}(s)\in  \left(\N/(\Gr,\vartheta_\Gr) \join \Np/(\Gr,\vartheta_\Gr)\right)_\St$ if and only if it is in at least one of $\left(\N/(\Gr,\vartheta_\Gr)\right)_\St$ or $\left(\Np/(\Gr,\vartheta_\Gr)\right)_\St$ which occurs if and only if $\vartheta_\Gr(s)$ is in one of the $G$-closed substructures $\N$ or $\Np$. 

For transitions, we note that a transition in $\left(\N/(\Gr,\vartheta_\Gr) \join \Np/(\Gr,\vartheta_\Gr)\right)_\Tr$ is entirely contained in one of  $\left(\N/(\Gr,\vartheta_\Gr)\right)_\Tr$ or $\left(\Np/(\Gr,\vartheta_\Gr)\right).$ Hence the transition must come from either $\N$ or $\Np$.
\end{proof}
}

As seen in Example \ref{ex:box_quo}, for a Kripke structure $\M$ with a group of state-mappings $\Gr$ of $\M$ and representative map $\vartheta_G$, it is possible for multiple $G$-closed substructures of $\M$ to map to the same substructure of the quotient structure $\M/(\Gr,\vartheta_G)$.



Recall that the join of $G$-closed substructures of $M$ is $G$-closed. 

\begin{definition}[$\Gr$-maximal]
Let $\M$ be a Kripke structure with group $G$ of state-mappings and representative map $\vartheta_G$. 
A $G$-closed substructure $\N$ of $\M$ is $\Gr$-maximal if 
\[
\N = \bigvee_{\substack{\Np \in \Lambda_{\M, G}\\ \Np/(G,\vartheta_G) \leq \N/(G,\vartheta_G)}} \Np.
\]
\end{definition}

\begin{lemma}\label{lem:G-maximal_join}
Let $\M$ be a Kripke structure and let $G$ be a group of state-mappings of $\M$ with representative map $\vartheta_G$. Let $\Mp,\Mpp$ be two $G$-maximal substructures of $\M$. Then $\Mp \join \Mpp$ is $G$-maximal. 
\end{lemma}

\condInc{
\begin{proof}
Consider a $G$-closed substructure $\Np$ such that $\Np/(G,\vartheta_G) \leq (\Mp\join \Mpp)/(G,\vartheta_G)$. Let $(s,t) \in \Np_\Tr$. Since $(\vartheta_G(s),\vartheta_G(t)) \in (\Mp\join \Mpp)/(G,\vartheta_G)_\Tr$ we conclude that $(\vartheta_G(s),\vartheta_G(t))$ is in either $\Mp/(G,\vartheta_G)$ or $\Mpp/(G,\vartheta_G)$. In either case, since $\Mp$ and $\Mpp$ are both $G$-maximal the transitions $(g(s),g(t))$ are already contained in $\Mp$ or $\Mpp$.
\end{proof}}
\begin{lemma}\label{lem:G-maximal_meet}
Let $\M$ be a Kripke structure and let $G$ be a group of state-mappings of $\M$ with representative map $\vartheta_G$. Let $\Mp,\Mpp$ be two $G$-maximal substructures of $\M$. Then $\Mp \meet \Mpp$ is $G$-maximal.
\end{lemma}
\condInc{
\begin{proof}
Suppose that $\N$ is the $G$-maximal substructure of $\M$ corresponding to $(\Mp\meet \Mpp)/(G,\vartheta_G)$. We note that $\N/(G,\vartheta_G) \leq \Mp/(G,\vartheta_G).$ Hence $\N \leq \Mp$. By symmetry $\N \leq \Mpp$ and thus $\N \leq \Mp \meet \Mpp.$
\end{proof}}

Combing Lemma \ref{lem:G-maximal_join} and \ref{lem:G-maximal_meet} allows us to make the following definition.

\begin{definition}[$G$-maximal lattice of substructures]
Let $\M$ be a Kripke structure and let $G$ be a group of state-mappings of $\M$ with representative map $\vartheta_G$. Then the set of $G$-maximal substructures of $\M$ forms a sublattice $\Lambda_{\M,\Gmax}$ of $\Lambda_\M$. 
\end{definition}

While in general the quotient map from $\Lambda_{\M,G}$ to $\Lambda_{\M/(G,\vartheta_G)}$ is always surjective, when restricted to $\Lambda_{\M,\Gmax}$ the map is injective and a lattice isomorphism.

\begin{theorem}[\textbf{$G$-Maximal Lattice Correspondence}]\label{thm:correspondence}
\sloppy Let $\mathcal{M}$ be a  Kripke structure and suppose that $G$ is a group of state-mappings of $\M$ with a representative map $\vartheta_G$. Then the map $f(\mathcal{N}) = \mathcal{N}/\left(G,\vartheta_G\right)$ is an isomorphism between the lattice of $G$-maximal substructures of $\mathcal{M}$ and the lattice of structures of $\mathcal{M}/\left(G,\vartheta_G\right)$. That is $f$ is an isomorphism from $\Lambda_{M,\Gmax}$ to $\Lambda_{M/(G,\vartheta_G)}.$
\end{theorem}

\condInc{
\begin{proof}
As noted the $G$-maximal lattice of substructures is in bijection with the lattice of substructures of $\M/(G,\vartheta_G)$. Since the meet and join of two $G$-maximal structures are determined entirely by the meet and join of their image in $M/(G,\vartheta_G)$, the correspondence is a lattice isomorphism. 
\end{proof}
}

So far, our correspondence between models $\mathcal{M}$ and $\mathcal{M}/\left(G,\vartheta_G)\right)$ has been entirely syntactical. In the next section we explore the semantic connections between $\mathcal{M}$ and $\mathcal{M}/\left(G,\vartheta_G\right)$. 

\subsection{Semantic relationships between structures and quotient structures.} 

Let $\Gr$ be a symmetry group of $\M$. 

\begin{definition} A CTL formula $f$ is \textbf{$\Gr$-invariant} over $\M$, if for every state $s$, every $g \in \Gr$, for all maximal propositional  subformulae $P$ of $f$, we have 
\[\M, s \sat P \iff \M, g(s) \sat P.\]
\end{definition}

\begin{lemma}\label{lem:eval_pre}
Let $\Gr$ be a group of state-mappings for a Kripke structure $\M$. If $f$ is $\Gr$-invariant, then the valuation of $f$ in $\overline{\mathcal{M}}= \mathcal{M}/\left(G, \vartheta_G\right)$ does not depend on the choice of representative map $\vartheta_G$.
\end{lemma}

\condInc{
\begin{proof}
Let $\vartheta_G, \varphi_G$ be two representative maps with respect to $G$. The quotients $\mathcal{M}/\left(G,\vartheta_G\right)$ and $\mathcal{M}/\left(G, \varphi_G\right)$ have corresponding transition maps. Moreover, for any maximal propositional subformula, $P$ of $f$ we have that 
\[
\M/\qmap,\vartheta_G(s) \sat P \iff 
\M/\left(\Gr,\varphi_{\Gr}\right), \varphi_G(s) \sat P.
\] 
The valuation of $f$ depends only on the valuation of its maximal propositional subformulae at all other states of the model. Since $\vartheta$ and $\varphi$ have the same evaluation of the maximal propositional subformulae at all states, the evaluation of $f$ at all states is the same. 
\end{proof}
}

This allows us to connect semantic statements about $\M$ with semantic statements about $\M/(G,\vartheta_G)$ for formulae that are $G$-invariant. 

If $G$ is a group of state-mappings of a structure $\M$ and $\vartheta_G$ is a representative map of $G$, then the path correspondence theorem establishes what we call a $G$-bisimulation between $\M$ and the structure $\M/(G,\vartheta_G)$.  

\begin{definition}[$\Gr$-bisimulation]
\label{def:qbisim}
Let $\M$ be a Kripke structure, $\Gr$ be a group of state-mappings of $\M$ with representative map $\rmap$, and $\Mq$ be $\M/(\Gr,\rmap)$.
Let $R \sub \MSt \times \MqSt$.
Then $R$ is a $\Gr$-bisimulation between $\M$ and $\Mq$ iff:
\begin{enumerate}
    \item \label{qbisim:state} If $(s,t) \in R$ then there exists $g \in \Gr$ such that $t = g(s)$, i.e., $t = \vartheta_G(s)$;
    
    \item \label{qbisim:transferM} If $(s,t) \in R$ and $s \ar s' \in \MTr$, then there exists $t'$ such that $t \ar t' \in \MqTr$ and $(s',t') \in R$ 
    
    
    \item \label{qbisim:transferMq} If $(s,t) \in R$ and $t \ar t' \in \MqTr$, then there exists $s'$ such that $s \ar s' \in \MTr$ and $(s',t') \in R$ 
    
\end{enumerate}
\end{definition}


\begin{lemma}
\label{lem:qbisim}
Let $\M$ be a Kripke structure, $\Gr$ be a group of state-mappings of $\M$ with representative map $\rmap$, and $\Mq$ be $\M/(\Gr,\rmap)$.
Let $\bsim = (s,\sq)$ where $s \in \MSt$ and $\sq = \rmap(s)$.
Then $\bsim$ is a $\Gr$-bisimulation.
\end{lemma}

\condInc{
\begin{proof}
We show that all clauses of Definition~\ref{def:qbisim} hold.
By definition of $\rmap$, $\sq = g(s)$ for some $g \in \Gr$, so Clause~\ref{qbisim:state} holds.
Theorem~\ref{thm:path-correspondence} shows the existence of the corresponding transitions needed to satisfy 
Clauses~\ref{qbisim:transferM} and \ref{qbisim:transferMq}
of Definition~\ref{def:qbisim}.
\end{proof}
}

%

\begin{lemma}
\label{lem:qbisimCTL}
Let $\M$ be a Kripke structure, $\Gr$ be a group of state-mappings of $\M$ with representative map $\rmap$, and $\Mq$ be $\M/(\Gr,\rmap)$. Let $s \in \MSt$, $t \in \Mq\St$, let $\bsim = (s,\overline{s}$ where $s\in \mathcal{M}_\St$ and $\overline{s} = \vartheta_G(s)$. If
$s \bsim t$. Let $f$ be a $\Gr$-invariant CTL formula. Then $\M, s \sat f \iff \Mq, t \sat f$.
\end{lemma}

\condInc{
\begin{proof}
Proof is by induction on the structure of CTL formulae, as given by the grammar in Section~\ref{sec:ctl}.

Since $f$ is $\Gr$-invariant over $\M$, then 
for all $t \in \M$ we have that 
$\M,t \sat P \mbox{ iff } \Mq,\rmap(t) \sat P$,
where $P$ is any maximal propositional subformula of $f$.

This provides the base case for the induction, where $f$ is an atomic proposition.

The induction step when $f$ has main operator one of $\neg, \land, \lor$ are straightforward, and omitted.
The induction steps when 
$f$ is one of $\AX g$, $\EX g$ and $\A[g \R h]$ use 
Definition~\ref{def:qbisim} to establish the existence of the necessary transitions/paths.
The details are standard, see for example 
Theorem 3.2 in \cite{BCG88}.
\end{proof}
}

%

\begin{lemma}
\label{lem:qequivCTL}
Let $\M$ be a Kripke structure, $\Gr$ be a group of state-mappings of $\M$ with representative map $\rmap$, and $\Mq$ be $\M/(\Gr,\rmap)$.
Let $f$ be a $\Gr$-invariant CTL formula and 
$s \in \MSt$. 
Then $\M,s \sat f \mbox{ iff } \Mq,\rmap(s) \sat f$.
\end{lemma}

\condInc{
\begin{proof}
By Lemma~\ref{lem:qbisim}, $s \bsim \rmap(s)$.
So, by Lemma~\ref{lem:qbisimCTL}, 
$\M,s \sat f \mbox{ iff } \Mq,\rmap(s) \sat f$.
\end{proof}
}

\begin{definition}[Repair]
Given a structure $\M$ and a formula $f$, we call a nonempty substructure $\Mp$ of $\M$ a \textbf{repair} of $\M$ with respect to $f$ if $\Mp \sat f$. 
\end{definition}

If a CTL formula $f$ is $G$-invariant, then the lattice correspondence will respect the valuation of $f$.

\begin{theorem}[\textbf{Repair Correspondence}]\label{thm:repair}
\sloppy Let $\M$ be a  Kripke structure, $\Gr$ be a group of state-mappings of $\M$ with a representative map $\rmap$, 
and $f$ be a $\Gr$-invariant CTL formula. 

Let $\N$ be a $\Gr$-closed substructure of $\M$, and 
let $\Nq = \N/\qmap$.
Then $\N, s \sat f \iff \Nq, \rmap(s) \sat f$.
\end{theorem}

\condInc{
\begin{proof}
We can show that $\N$ and $\Nq$ are $\Gr$-bisimilar 
under the usual relation of each state $s$ in $\N_\St$ with 
$\rmap(s)$ in $\Nq_\St$. 
%
%
%
\end{proof}
}

It should be noted however that there are formulae $f$ and  structures $\M$, such that $f$ is $\Gr$-invariant for some symmetry group $\Gr$ of $\M$ and there are no $\Gr$-closed repairs of $f$, but there are repairs that are not $\Gr$-closed. That is, not every repairable structure has a symmetric repair. 

Theorem \ref{thm:repair} gives us the commuting diagram in Figure \ref{fig:com_diagram}. 

\begin{wrapfigure}{r}{2.5in}
    \centering
\includegraphics[scale= .4]{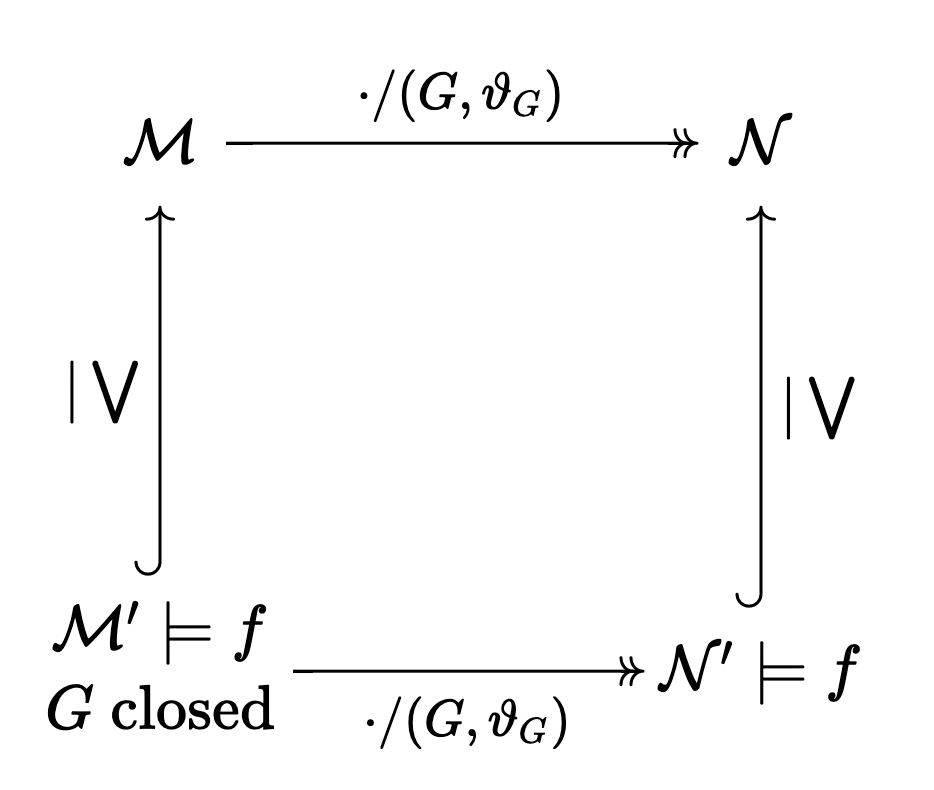}   \caption{Commutative diagram illustrating the connection between a structure $\M$, a quotient structure $\Mq$, and substructures of $\Mq$ satisfying a formula $f$ with $G$-closed substructures of $\M$ satisfying $f$. }
    \label{fig:my_label}
    \label{fig:com_diagram}
\end{wrapfigure}


\section{Repair of concurrent programs}

We now apply the repair correspondence theorem to realize our goal: the repair of concurrent programs. 
A concurrent program $P = P_1 \pl \ldots \pl P_n$ consists of $n$ sequential processes executing in parallel. Each process $P_i$ 
is a set of \emph{$i$-actions}, \ie tuples 
$(s_i, B \ar A, t_i)$, where $s_i, t_i$ are \emph{local states} of $P_i$, $B$ is a guard, and $A$ is an assignment to shared variables. 
We say \emph{action} when we wish to ignore the identity of the process that executes the action.
Each $P_i$ has a local set of atomic propositions $\AP_i$ which it alone can modify, but other processes can read. Each local state of $P_i$ determines as assignment of truth values to the atomic propositions in $\AP_i$. We set $\AP = \AP_1 \un \ldots \un \AP_n$. 
A \emph{global state} $s = (s_1,\ldots,s_n, v_1,\ldots,v_m)$ of $P$ is a tuple consisting of local states $s_1,\ldots,s_n$
for processes $P_1, \ldots, P_n$ respectively, and values
$v_1,\ldots,v_m$ for shared variables $x_1,\ldots,x_m$.
Define $s \pj i = s_i$ for $i=1,\ldots,n$ 
and $s(x_j) = v_j$ for $j=1,\ldots,m$.
The truth values that atomic propositions in $\AP_i$ have in $s$ must be the same as those in $s_i$. 
The states of Kripke structures that we have treated above are global states as described here.

Given a concurrent program $P$ and a CTL formula $f$, we wish to modify $P$ to produce a $P'$ such that 
$\M \sat f$, where $\M$ is the global state transition graph of $P'$. The modification is "subtractive", that is, it only removes behaviors, it does not add behaviors (which would be, in essence, a synthesis, and not a repair, algorithm).
We proceed as follows:

\begin{enumerate}
\item \label{repair:generateRed} Generate the symmetry-reduced state transition graph $\N$ of $P$. We use the 
algorithm given by Figure 1 in \cite{EmSi96}, which is reproduced in Appendix ?.

\item \label{repair:repairRed} Repair $\N$ \wrt the given CTL specification $f$. Let $\Np$ be the repaired structure, so that $\Np \sat f$.

\item \label{repair:extractProg} Extract a repaired concurrent program $P'$ from $\Np$
\end{enumerate}

Steps 2 and 3 are given in the subsections below.

\subsection{Repair of symmetry-reduced structures}

We apply the model repair algorithm of \cite{ABS18} 
to the symmetry-reduced state transition graph $\N$ of $\Prg$, and the specification $f$ of $\Prg$. This algorithm is sound and complete, so that if $\N$ has some substructure that satisfies $f$, then the algorithm will return such a substructure $\NR$. If not, the algorithm will report that no repair exists, in which case there is no 
subtractive modification of $\Prg$ which satisfies $f$.

\subsection{Extraction of concurrent programs from symmetry-reduced structures}
\label{sec:prog-extraction}

\newcommand{\bigland}{\bigwedge}
\newcommand{\stof}[1]{\{\hspace*{-0.3em}|#1|\hspace*{-0.3em}\}}
\newcommand{\dn}{\mbox{$\hspace{0em}\downarrow\hspace{0em}$}}

\newcommand{\SH}{\mbox{$\cal SH$}}  
\newcommand{\lfalse}{\mathit{false}}
\newcommand{\ltrue}{\mathit{true}}

We extract concurrent program $\PrgR$ from $\NR$ following \cite{EC82}, by projecting each transition in $s \la{i} t \in \NTr$ onto the process $P_i$ making the transition. 
For each $s \la{i} t \in \NTr$ we define 
$\action{s \la{i} t}$ to be the 
$i$-action $(s \pj i, B \ar A, t \pj i)$ as follows.
The start, end states of this action are the projections of $s, t$ respectively onto $P_i$. 
%
We define the guard 
$B  =  
``(\AND_{Q \in \NLb(s)} Q) \land
  (\AND_{Q \not\in \NLb(s)} \neg Q) \land
  (\AND_{x \in \SH} x = s(x))\mbox{\rm ''}
$
where $Q$ ranges over $\AP - \AP_i$. 
%
When $\Prc_i$ is in local state $s_i$, the guard $B$ checks that the current global state is actually $s$.
%
We define the assignment $A = `` \PA x := t(x) \mbox{\rm ''}$
where $x$ ranges over all shared variables such that 
$t(x) \ne s(x)$. That is, $A$ makes the necessary updates to the shared variables in moving from state $s$ to state $t$.

Let $\Ac = \set{ \action{s \la{i} t} \stt s \la{i} t \in \NTr}$,
\ie $\Ac$ is the set of actions generated by all the transitions of $\N$. 
%
%
In principle, each process $\PrcR_i$ of the repaired program $\PrgR$ can be set to be the set of $i$-actions in $\Ac$. It is straightforward to show that $\PrgR$ thus constituted generates exactly $\N$ as its global state transition diagram. Since $\NR \sat f$, we also have $\PrgR \sat f$ as desired.
However, even though $\NR$ is total, it is possible for some processes in $\PrgR$ to have "dead end" local states from which they have no outgoing actions. 
We therefore repeatedly apply some $g \in \Gr$ to $\Ac$, until all such "dead end" local states acquire at least one outgoing action. 
One can continue this application of $g \in \Gr$ to $\Ac$ until no more actions are generated, which produces a "$\Gr$-closed" program $\PrgR$. This is the "largest" correct program that can be produced as a repair of $\Prg$. Using the largest program may increase the available concurrency, but may produce a large number of actions, and so increase the complexity of the repair procedure.


\begin{theorem} 
\label{thm:correctRepair}
Let $\NR$ be a repaired Kripke structure as above, $\PrgR$ be a concurrent program extracted from $\NR$ as above, 
and $\Mp$ be the global state transition diagram of $\PrgR$. Then  every initial state of $\Mp$ is $\Gr$-bisimilar to some initial state of $\NR$, and vice-versa.
\end{theorem}
\condInc{
\begin{proof}
Let $\bsim$ consist of the pairs $(s,t)$ such that $s \in \NRSt$, $t \in \MpSt$, and $t = g(s)$ for some $g \in \Gr$.
The proof that $\bsim$ is a $\Gr$-bisimulation
follows from the definitions of program extraction and 
program execution semantics.
By the definitions of program extraction and 
program execution semantics, $\NR$ and $\Mp$ have isomorphic sets of initial states, with corresponding states having the same labels. It follows that these states are bisimilar.
Hence $\bsim$ is a $\Gr$-bisimulation which moreover relates each initial state of $\NR$ to an initial state of $\Mp$ and vice-versa. 
\end{proof}
}

\begin{lemma}
\label{lem:correctRepair}
Let $\NR$ be a Kripke structure as above that is repaired 
\wrt a $\Gr$-invariant CTL formula $f$. Let  $\PrgR$ be a concurrent program extracted from $\NR$ as above, and $\Mp$ be the global state transition diagram of $\PrgR$. Then $\Mp \sat f$.
\end{lemma}
\condInc{
\begin{proof}
Since $\NR$ is repaired \wrt $f$, we have $\NR \sat f$ by 
Corollary 4.5 of \cite{ABS18}.
Hence $\NR, s_0 \sat f$ for every initial state $s_0$ of $\NR$. 
By Theorem~\ref{thm:correctRepair}, every initial state of $\Mp$ is $\Gr$-bisimilar to some initial state of $\NR$.
Hence by Lemma~\ref{lem:qbisimCTL}, 
$\Mp, t_0 \sat f$ for every initial state $t_0$ of $\Mp$. 
Hence $\Mp \sat f$.
\end{proof}
}

\section {Examples}

\subsection{Two process mutual exclusion} \label{lattice-example}
\label{sec:mutex2}

Figure~\ref{fig:mutex2} shows a Kripke structure $\M$ for 
mutual exclusion of two processes $P_1$ and $P_2$.
$P_i$ ($i=1,2$) can be in any of three local states: $Ni$ (neutral, doing local computation), $Ti$ (trying, has requested critical section entry), and $Ci$ (in the critical section). 
$\M$ has exactly two symmetries: the identity map, and the map that swaps process indices $1$ and $2$. We construct \emph{the} quotient model under this symmetry group. 
The reduced structure $\Mq = \M/\qmap$ is also shown in
Figure~\ref{fig:mutex2}.
Transitions of $P_1$, $P_2$ are shown in blue, red, respectively. $\Mq$ has a transition (shown in black) from state \textbf{S6} to \textbf{S1}, 
which is the quotient of the transition from \textbf{S6} to \textbf{S2} in $\M$, i.e., $\vartheta_G(S6) = S6$ and $\vartheta_G(S2) = S1$ so the edge $(\vartheta_G(S6), \vartheta_G(S1))$ occurs in $\Mq$.

Consider the formula $\varphi = \AG \neg(C_1 \land C_2)$. Clearly, $\M \not\sat \varphi$. Since $\varphi$ is $\Gr$-invariant, this is equivalent to $\Mq \not\models \varphi$ for any representative map $\vartheta_G$. From the correspondence theorem, any substructure of $\Mq$ that satisfies $\varphi$ corresponds to a $\Gr$-maximal substructure of $\M$ that satisfies $\varphi$. 

Figure~\ref{fig:mutex2repaired} shows the repair of the reduced structure $\Mq$, and the resultant lifting of 
the repair to $\M$. The deleted transitions and states 
are shown dashed.

\begin{figure}
\begin{center}
\includegraphics[width=8cm]{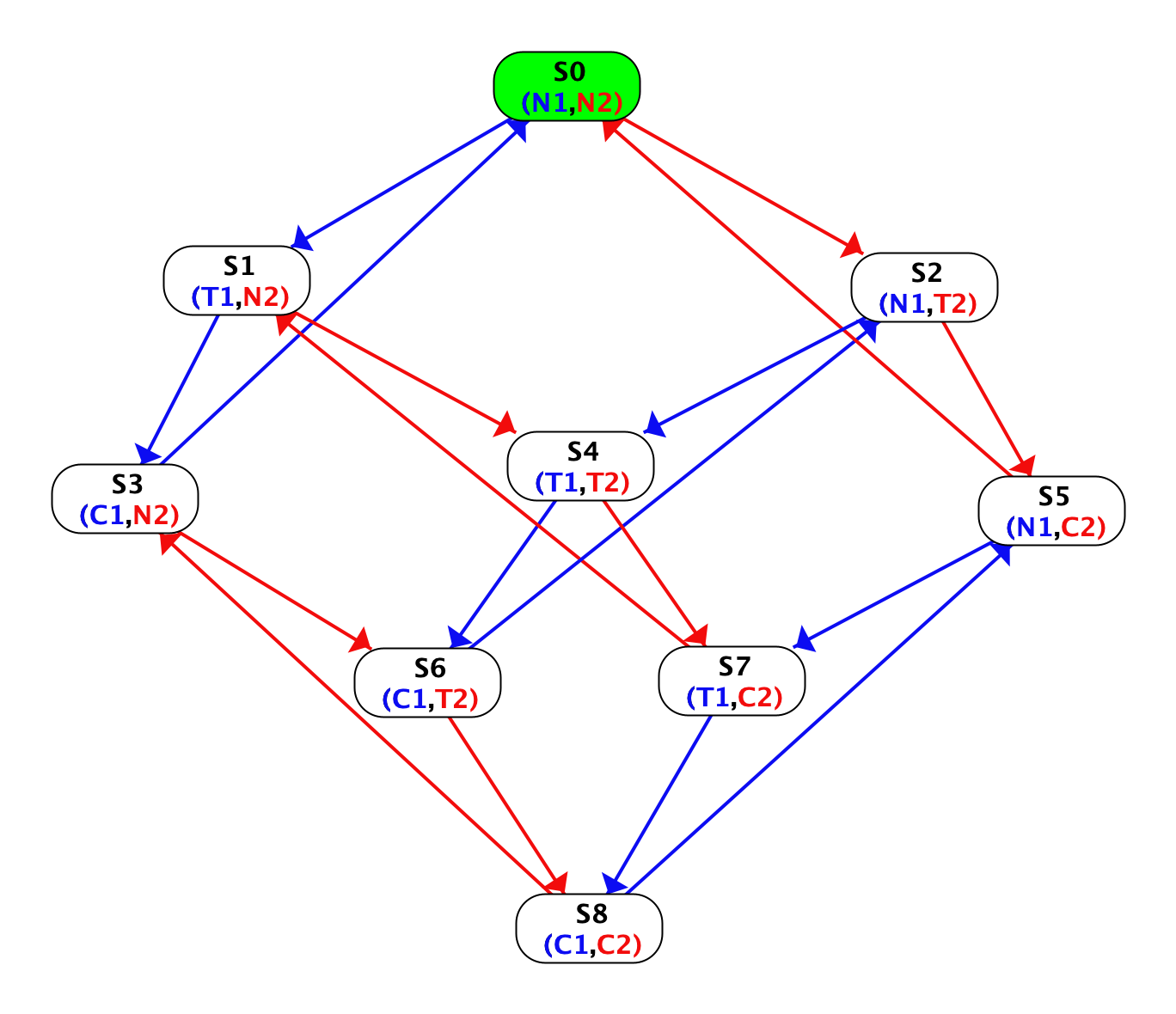}
\includegraphics[width=5cm]{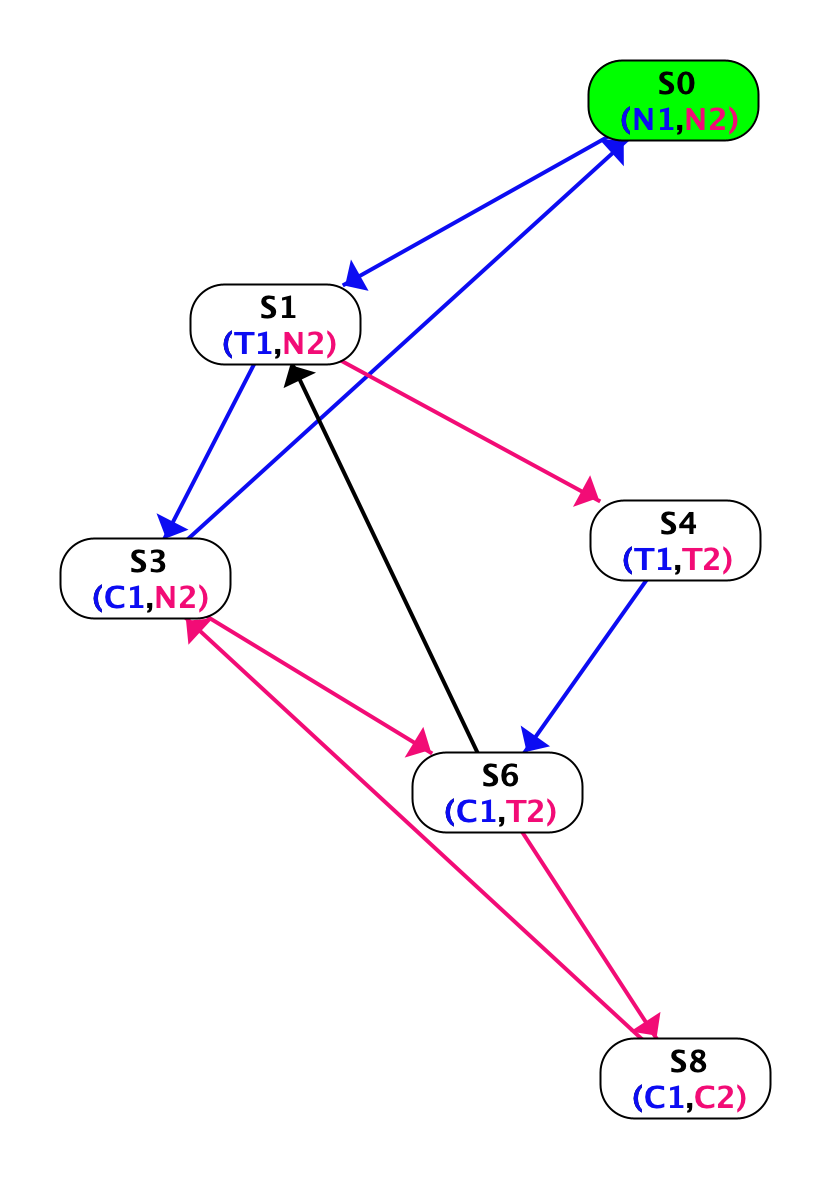}
\caption{The original model $\M$ and quotient $\Mq = \M/\qmap$.}
\label{fig:mutex2}
\end{center}
\end{figure}

\begin{figure}
\begin{center}
\includegraphics[width=5cm]{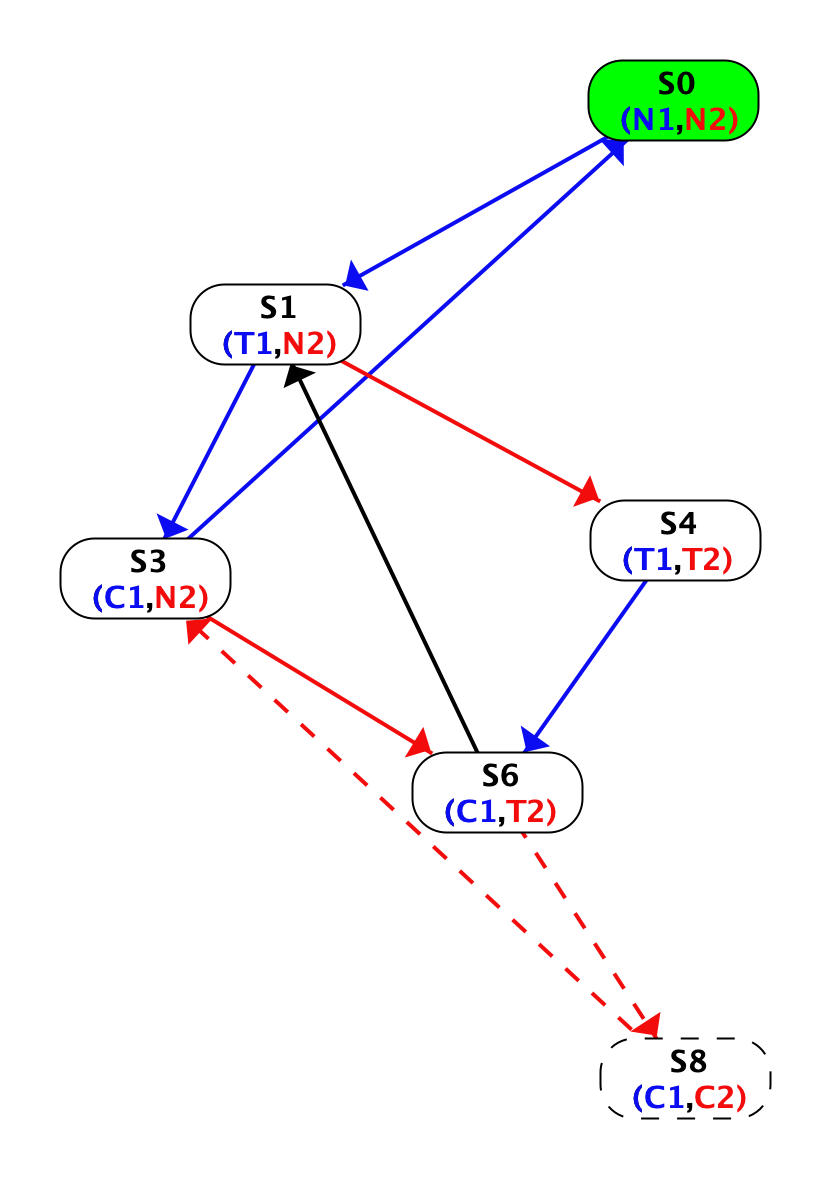}
\includegraphics[width=8cm]{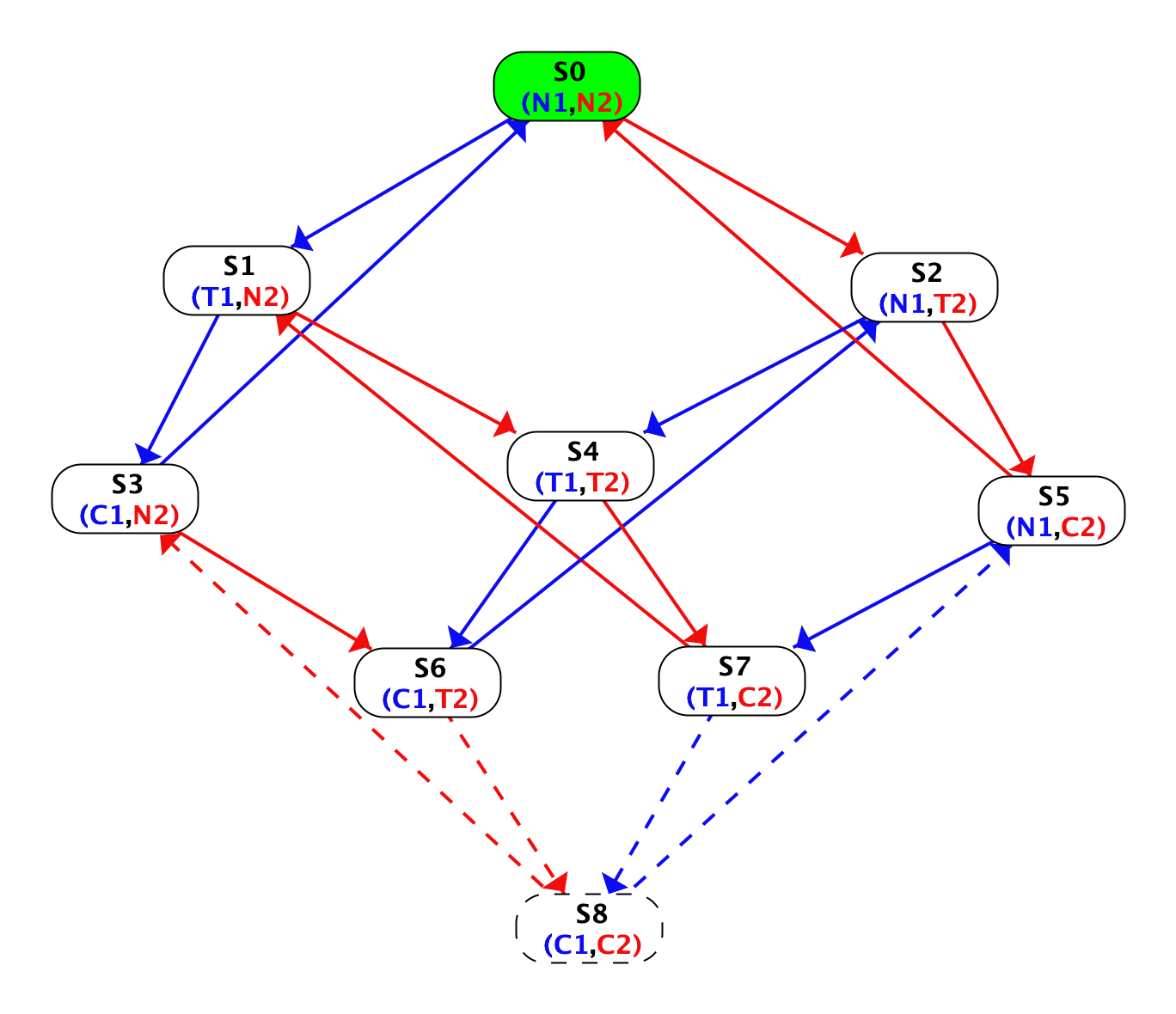}
\caption{The repair of $\Mq$ and the lifting of the repair to $\M$.}
\label{fig:mutex2repaired}
\end{center}
\end{figure}

\subsection{$N$-Process Mutual Exclusion}
\label{sec:mutex3}

We now consider mutual exclusion for $N$-processes. To reduce clutter, we remove the trying $(Ti)$ state, and we give a concrete example for 3 processes --- the generalization to $N$ processes is straightforward.
We start with the na\"ive concurrent program shown in Figure~\ref{fig:mutex3prog}, where actions are labeled with the trivial guard $\true$. The resulting Kripke structure is an $N$-dimensional hypercube, 
and Figure~\ref{fig:mutex3} shows the structure $\M$ for 3-processes, along with the quotient $\Mq = \M/\qmap$.
We consider the mutual exclusion specification 
$\AND_{i \ne j} \AG \neg(C_i \land C_j)$.
The symmetry group $\Gr$ for both structure and specification is the full permutation group on the indices $\set{1,\ldots,N}$. For $N$-processes, we have that the quotient model by the full group of symmetries has $N+1$ states, while the original model would have $2^N$ states.
Figure~\ref{fig:mutex3repaired} shows the repair of the quotient $\Mq$ and then the lifting of the repair to the original structure $\M$. 
Figure~\ref{fig:mutex3progrep} shows the correct (repaired) program $\PrgR$ that is extracted from the repaired quotient in Figure~\ref{fig:mutex3repaired}.
For $N$ processes, the guard on actions of $\PrcR_i$ is $\AND_{j \ne i} N_j$.


\begin{figure}
\begin{center}
\includegraphics[width=15cm]{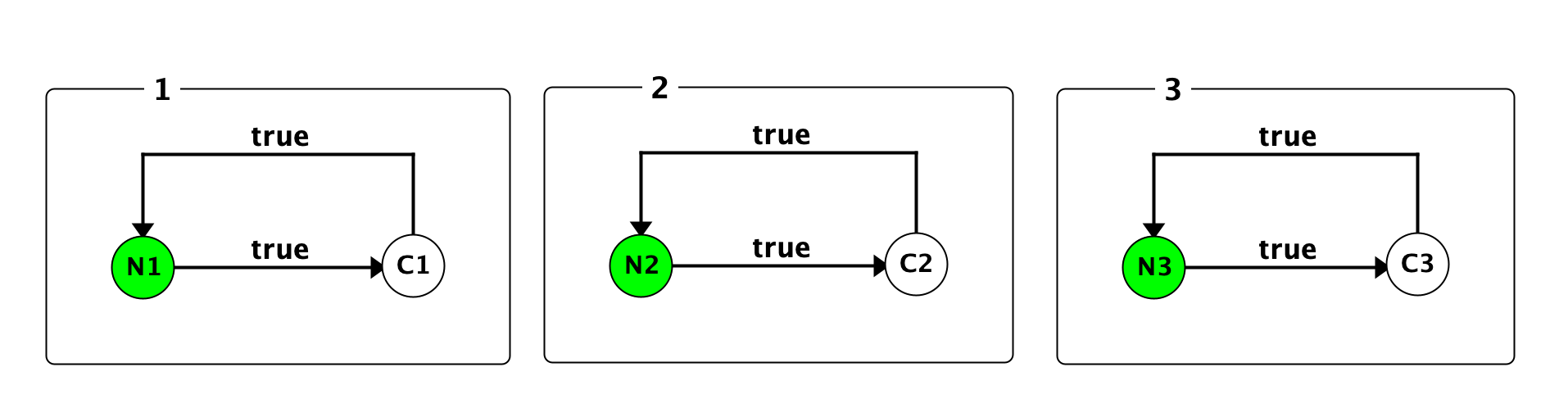}
\caption{The original program $\Prg$}
\label{fig:mutex3prog}
\end{center}
\end{figure}

\begin{figure}
\begin{center}
\includegraphics[width=8cm]{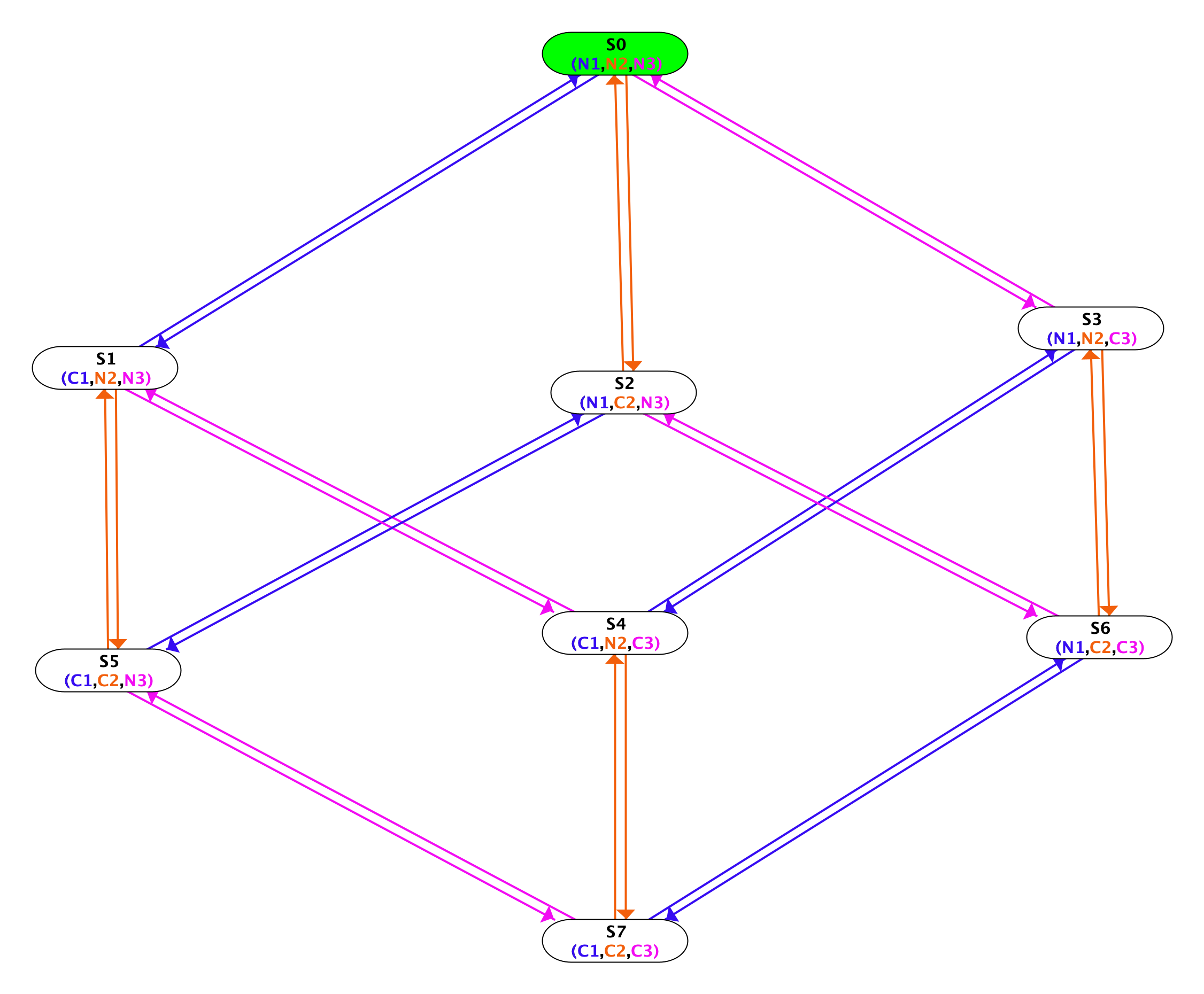}
\includegraphics[width=4.5cm]{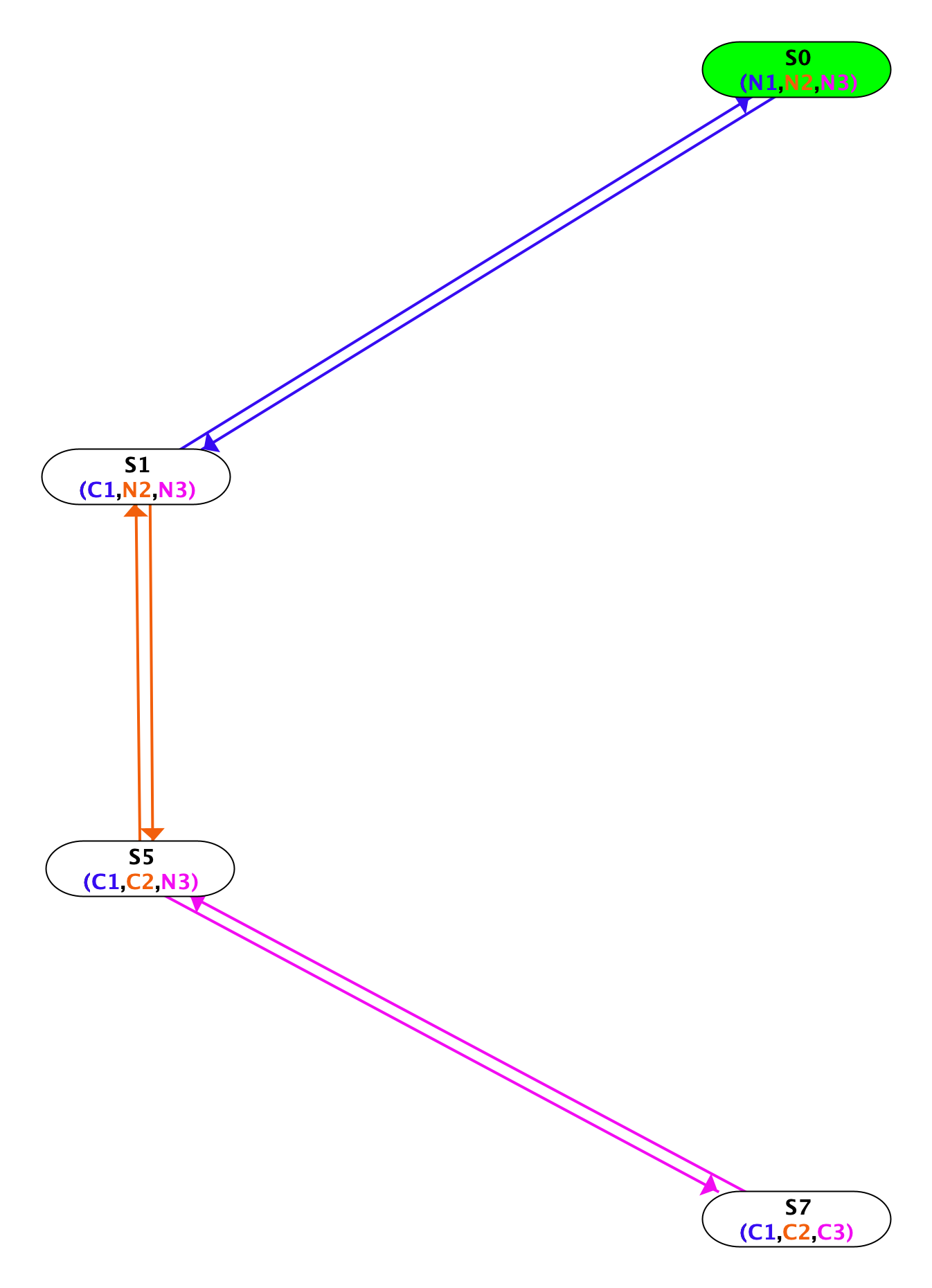}
\caption{The original model $\M$ and quotient $\Mq = \M/\qmap$.}
\label{fig:mutex3}
\end{center}
\end{figure}

\begin{figure}
\begin{center}
\includegraphics[width=4.5cm]{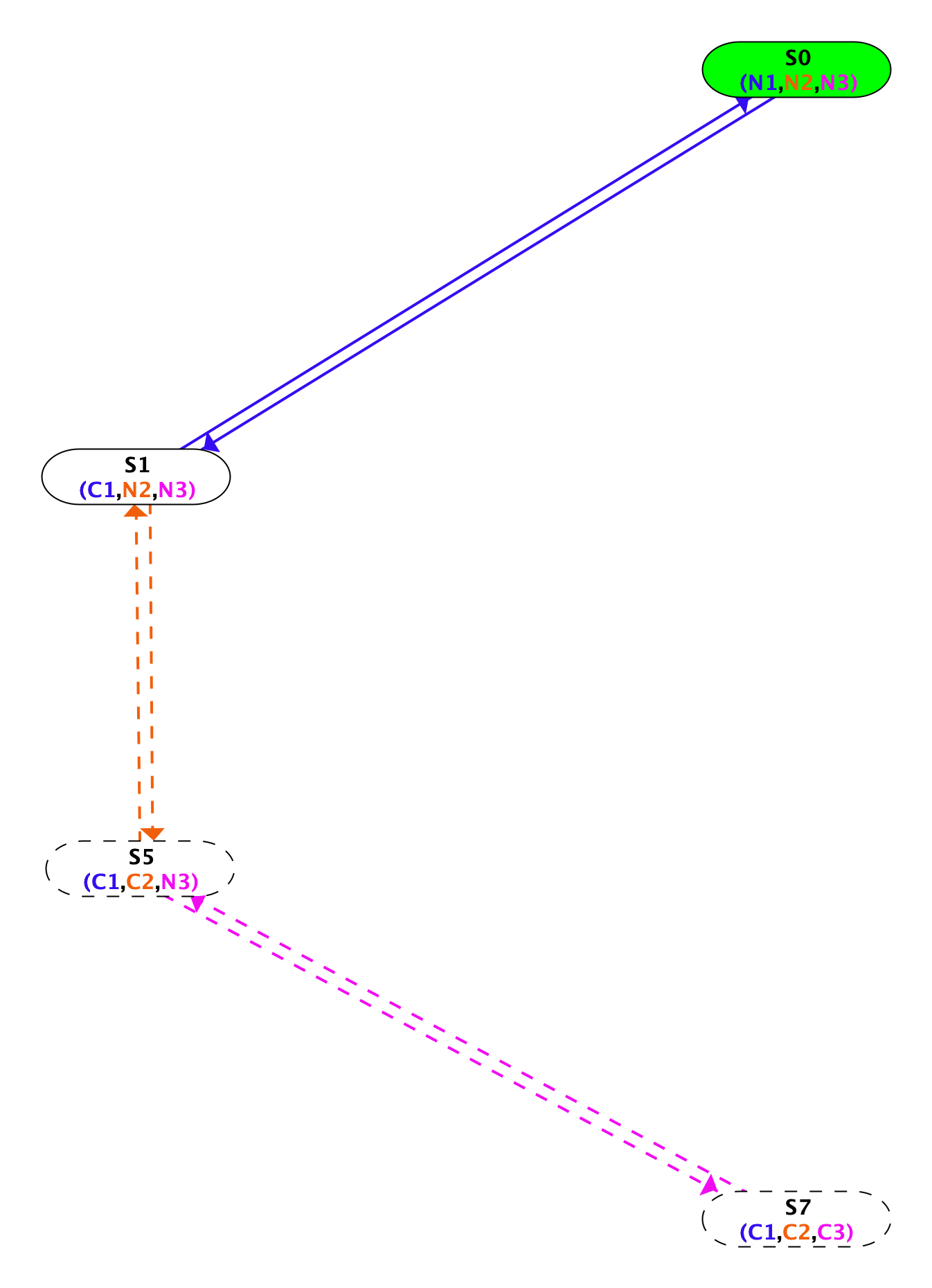}
\includegraphics[width=8cm]{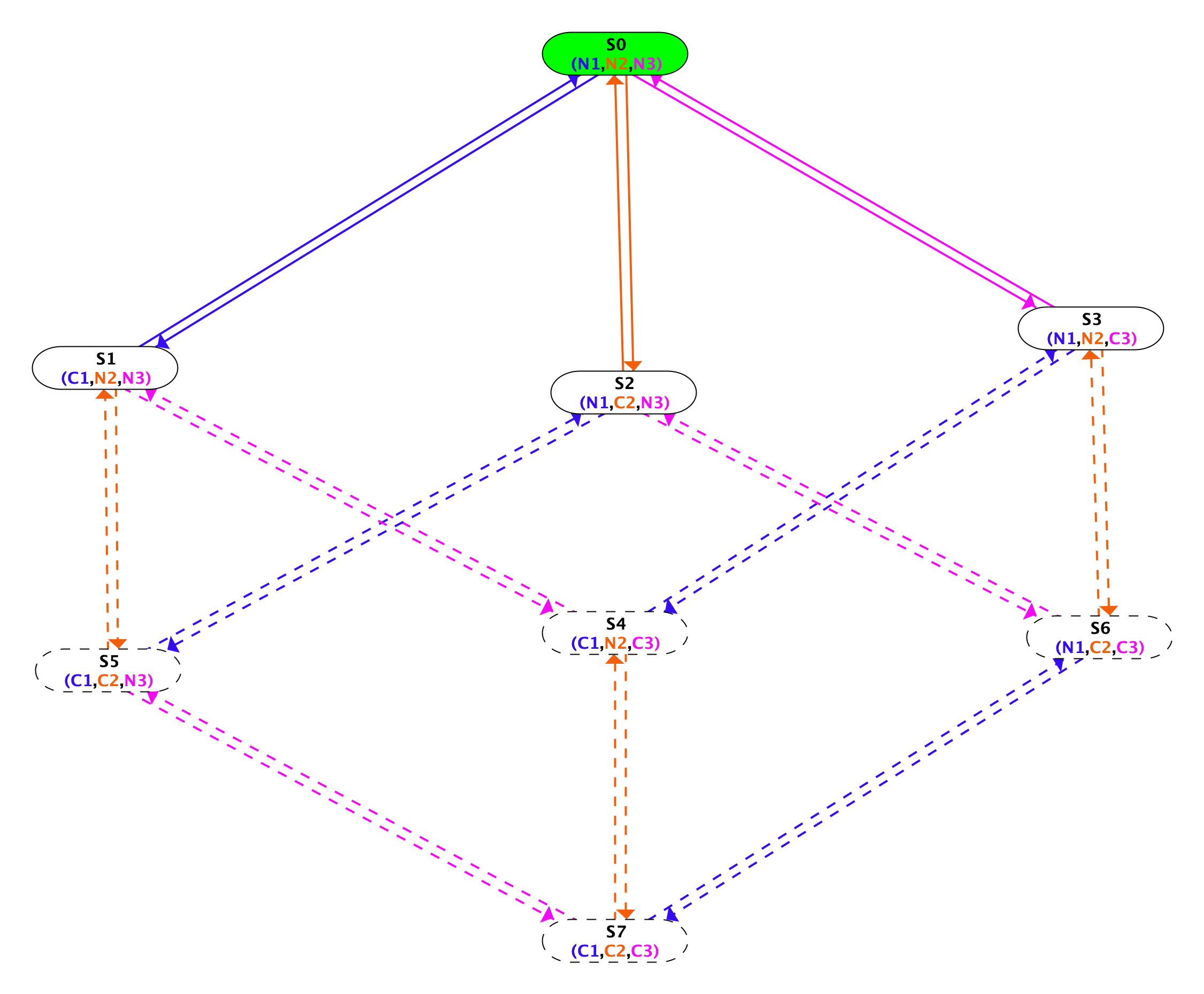}
\caption{The repair of $\Mq$ and the lifting of the repair to $\M$.}
\label{fig:mutex3repaired}
\end{center}
\end{figure}

\begin{figure}
\begin{center}
\includegraphics[width=15cm]{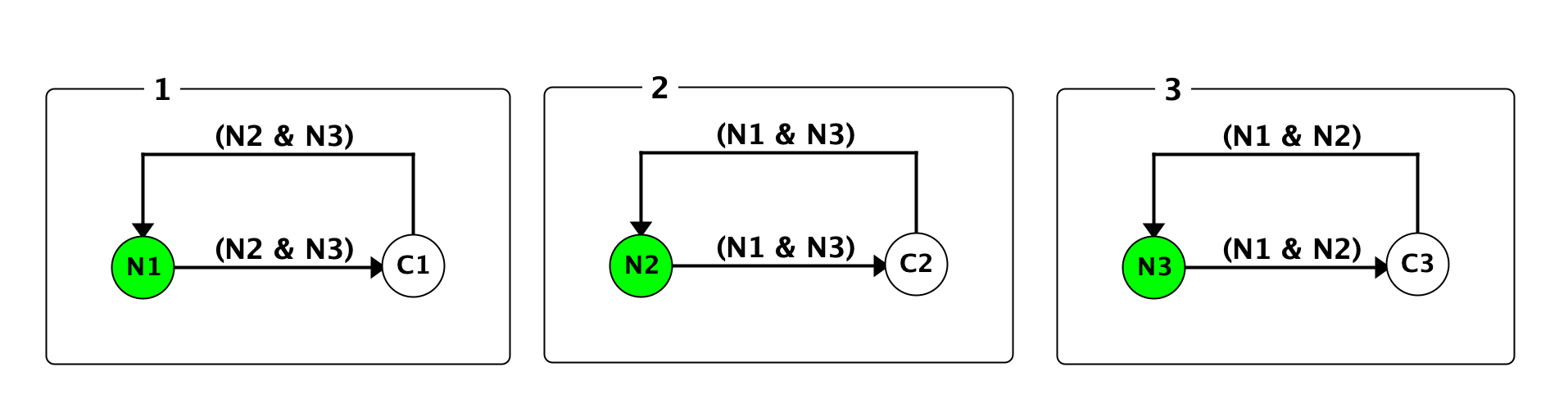}
\caption{The repaired program $\PrgR$}
\label{fig:mutex3progrep}
\end{center}
\end{figure}

\section{Conclusions}

We presented a theory for the substructures of a given Kripke structure $\M$,  their organization into lattices, and how these substructures can interact with a group of state-mappings of $\M$. We showed the existence of an isomorphism between the lattice $\Lambda_{M,\Gmax}$ of $G$-maximal substructures of $\M$ and the lattice of substructures of $\M\qmap$. Using this theory we were able to devise a method for repairing concurrent programs which exploits symmetry. 

\clearpage
\bibliography{Emerson,refs,repair}

\newpage
\appendix
\section{CTL Semantics}
\label{app:ctl}

\begin{definition}
\label{def:kripke_semantics}
\label{defn:kripkeSemantics}
$\M,s \sat \vphi$ means that formula $\vphi$ is true  in state $s$ of structure $M$, and
$\M,s \not\sat \vphi$ means that $\vphi$ is false in $s$.
\begin{enumerate}
\item \label{clause:kripkeSemantics:true} $\M,s \sat \true$ and $\M,s \not\sat \false$
\item \label{clause:kripkeSemantics:lab} $\M,s \sat p  \text{ iff } p \in L(s)$ where atomic proposition $p \in \AP$

\item \label{clause:kripkeSemantics:neg} $\M,s \sat \neg \vphi$  \text{ iff } $\M,s \not\sat \vphi$
\item \label{clause:kripkeSemantics:and} $M,s \sat \vphi \land \psi
           \text{ iff } \M,s \sat \vphi \text{ and } M,s \sat \psi$
\item \label{clause:kripkeSemantics:or} $M,s \sat \vphi \lor \psi
           \text{ iff } \M,s \sat \vphi \text{ or } M,s \sat \psi$

\item \label{clause:kripkeSemantics:AX} $\M,s \sat \AX \vphi \text{ iff }
            \text{ for all $t$ such that } (s,t) \in T: (\M,t) \sat \vphi$
\item \label{clause:kripkeSemantics:EX} $M,s \sat \EX \vphi \text{ iff } 
            \text{ there exists $t$ such that } (s,t) \in T \text{ and }  (\M,t) \sat \vphi$

\item \label{clause:kripkeSemantics:AV} $\M,s \sat \A[\vphi \V \psi] \text{ iff }   
           \text{ for all fullpaths $\pi = s_0,s_1,\ldots$  starting from $s=s_0$: }$\\
\ind $\forall k \geq 0:   (\forall j < k: \M,s_j \not\sat \vphi ) \text{ implies } \M,s_k \sat \psi$

\item \label{clause:kripkeSemantics:EV} $\M,s \sat \E[\vphi \V \psi] \text{ iff } 
           \text{ for some fullpath $\pi = s_0,s_1,\ldots$  starting from $s=s_0$: }$\\
\ind $\forall k \geq 0: (\forall j < k: \M,s_j \not\sat \vphi) \text{ implies } \M,s_k \sat \psi$
\end{enumerate}
\end{definition}

\end{document}